\documentclass[lettersize,journal]{IEEEtran}
\IEEEoverridecommandlockouts
\usepackage[utf8]{inputenc}
\usepackage{color,setspace}
\usepackage{cite}
\usepackage{amssymb,amsfonts}
\usepackage{enumerate}
\usepackage{bbm}
\usepackage{graphicx}
\usepackage{epstopdf}
\usepackage{subfigure}
\usepackage{stfloats}
\usepackage[cmex10]{amsmath}
\usepackage{tasks}
\usepackage{enumitem}
\usepackage{bm}
\usepackage{xcolor}
\usepackage{xspace}
\usepackage{colortbl}
\usepackage{amsthm}
\usepackage{amsmath}
\usepackage{algpseudocode}
\usepackage[ruled,lined,commentsnumbered]{algorithm2e}

\def\BibTeX{{\rm B\kern-.05em{\sc i\kern-.025em b}\kern-.08em
T\kern-.1667em\lower.7ex\hbox{E}\kern-.125emX}}

\newtheorem{theorem}{\textit{Theorem}}
\newtheorem{lemma}{{\textit{Lemma}}}
\newtheorem{proposition}{\textit{Proposition}}
\newtheorem{corollary}{\textit{Corollary}}
\newtheorem{definition}{\textit{Definition}}
\newtheorem{remark}{\textit{Remark}}

\title{Optimizing AoI at Query in Multiuser Wireless Uplink Networks: A Whittle Index Approach}
\author{\IEEEauthorblockN{Jingwei Liu and He (Henry) Chen}
\thanks{The work of J. Liu is supported in part by the Innovation and Technology Fund (ITF) under Project ITS/204/20 and the CUHK direct grant for research under Project 4055166. The work of H. Chen is supported in part by RGC
General Research Funds (GRF) under Project 14205020.

J. Liu and H. Chen are with Department of Information Engineering, The Chinese University of Hong Kong, Hong Kong SAR, China 
(Email: \{lj020, he.chen\}@ie.cuhk.edu.hk).
}
\thanks{The authors thank Qian Wang for her useful discussions on the network model.}
}
\date{September 2024}

\begin{document}

\maketitle
\begin{abstract}
In this paper, we explore how to schedule multiple users to optimize information freshness in a pull-based wireless network, where the status updates from users are requested by randomly arriving queries at the destination. We use the age of information at query (QAoI) to characterize the performance of information freshness.
Such a decision-making problem is naturally modeled as a Markov decision process (MDP), which, however, is prohibitively high to be solved optimally by the standard method due to the curse of dimensionality. To address this issue, we employ Whittle index approach, which allows us to decouple the original MDP into multiple sub-MDPs by relaxing the scheduling constraints. However, the binary Markovian query arrival process results in a bi-dimensional state and complex state transitions within each sub-MDP, making it challenging to verify Whittle indexability using conventional methods. After a thorough analysis of the sub-MDP's structure, we show that it is unichain and its optimal policy follows a threshold-type structure.
This facilitates the verification of Whittle indexability of the sub-MDP by employing an easy-to-verify condition.
Subsequently, the steady-state probability distributions of the sub-MDP under different threshold-type policies are analyzed, constituting the analytical expressions of different Whittle indices in terms of the expected average QAoI and scheduling time of the sub-MDP.
Building on these, we devise an efficient algorithm to calculate Whittle indices for the formulated sub-MDPs.
The simulation results validate our analyses and show the proposed Whittle index policy outperforms baseline policies and achieves near-optimal performance.


\end{abstract}
\section{Introduction}

In recent years, the widespread adoption of time-critical wireless communication systems has drawn significant attention to the information freshness \cite{7467436,7799033,6917404}.
As a result, the age of information (AoI) metric, defined as the time elapsed since the generation of the most recently received message at the destination, has been proposed to quantify information freshness, see e.g., \cite{kosta2017age,sun2022age,5984917,6195689,10.1145/3323679.3326520,8514816} and references therein.
Most studies focused on the AoI metric have primarily investigated the minimization of the long-term average AoI across various network settings (e.g., see \cite{8486307,8406846,8845182,8000687,8514816,8935400,8845083,8807257}).
That is, these works assumed a \textit{push-based} wireless communication system, where monitors at the destination continuously request the most recently received information.
In contrast, the information freshness in \textit{pull-based} wireless networks has obtained relatively less consideration.
Specifically, at the destinations in the pull-based communication systems, a query arrival process occurs at each monitor, and the received messages are sampled and used by the monitor only at specific time instants when queries arrive, such as satellite Internet of Things (IoT) networks and ecological monitoring systems \cite{9149009,8935395,5597912,10.1145/1689239.1689247}.
To more effectively quantify information freshness in pull-based communication systems, the Age of Information at Query (QAoI) metric has been introduced and increasingly explored in the literature, for instance, in\cite{9676636,9834897,10376462,10552992,10049450} and references therein.
In particular, QAoI is defined as the instantaneous AoI values when the queries arrive at the destination.

There have been some research efforts dedicated to addressing the transmission scheduling problems to optimize the long-term average QAoI in different network scenarios. The QAoI-oriented transmission scheduling coordinations for single-user networks were researched in \cite{9676636,9834897,10376462}, in which when to transmit the status update of the single user is guided by the query arrivals.
Efforts on solving the QAoI-based scheduling problem in wireless two-user uplink networks have also been made in the open literature.
In such a scenario, monitors with query arrivals are deployed at a base station (BS), which schedules the transmissions of status update packets from the end devices corresponding to the monitors in the uplink.

In \cite{10552992}, the authors investigated the QAoI optimization of a reconfigurable intelligent surface-assisted wireless network, including two energy harvesting sensors and a monitor.
The optimization problem was approximated as a convex problem, which was solved by an alternating optimization approach.
However, this study cannot be applied to networks with more than two sensors.
Zakeri et al. in \cite{10049450} considered a time-slotted heterogeneous status update system with transmission and sampling frequency limitations with one source in the ``stochastic arrival'' model, one source in the ``generate-at-will'' model, and a transmitter deploying a buffer-aided and a monitor.
In each time slot, the transmitter decides whether the generate-at-will source is sampled by its transmitter and schedules the transmission of the two sources to minimize the QAoI performance of the network.
After modeling the decision-making problem as a constrained Markov decision process (CMDP), the authors selected to optimize the long-term discounted network-wide QAoI reward, since the CMDP is not unichain, by solving the corresponding linear programming (LP) of the CMDP.
Then, a reduced-complexity heuristic policy also built upon the LP approach was developed by investigating a weakly coupled CMDP problem relaxed from the original CMDP.
Although the proposed heuristic policy was shown applicable to networks including multiple stochastic arrival sources, this method suffers from high computational complexity with a large number of sources.
This is because of the increase in the number of variables and constraints in the LP associated with the weakly coupled CMDP, which grows with the number of sources, leading to an exponential rise in the complexity of solving this LP.
To the best of our knowledge, an efficient scheduling scheme for optimizing the long-term average QAoI in wireless multiuser uplink networks has not been thoroughly explored in the literature.

To fill this gap, in this paper we consider a multiuser scheduling problem for minimizing the expected sum QAoI of a wireless uplink network, where multiple status-updating end devices are scheduled to transmit their updates to the corresponding monitors deployed at a base station (BS). Moreover, each monitor experiences stochastic query arrivals, which follow a binary Markovian process. 
In this context, the most recent status update packets received by the monitor are requested when queries arrive.
We model the considered QAoI-oriented scheduling problem as an MDP problem.
The network-wide instantaneous AoI and query arrival situations are jointly defined as the state of the formulated MDP.
We summarize the main contributions of this paper as follows.
\begin{itemize}
    \item As the number of end devices increases, the state space grows exponentially, making it challenging to compute the optimal scheduling policy of the formulated MDP using standard methods due to the curse of dimensionality \cite{papadimitriou1987complexity}. Hence, we adopt the Whittle index method \cite{Whittle_1988} by regarding the formulated MDP as a restless multi-armed bandit (RMAB). Specifically, we relax the scheduling constraint of the formulated MDP. This allows us to decouple the original problem into multiple constrained problems, which are then formulated as independent sub-MDPs using Lagrange relaxation.
    \item Through an in-depth analysis of the sub-MDP structure, we rigorously prove that the sub-MDP is unichain, and its optimal scheduling policy is threshold-type with two thresholds associated with the query arrival state. Building on this, we establish the Whittle indexability of the sub-MDP by verifying the active condition of indexability, a recently introduced and easy-to-verify criterion \cite{zhoueasier}. 
    To obtain more design insights, we investigate the steady-state probability distributions of the formulated sub-MDP under the threshold-type policies in a variety of contexts.
    On this basis, we derive the analytical expressions for different Whittle indices in terms of the expected average QAoI and scheduling time of the sub-MDP under threshold-type policies.
    
    \item 
We design an efficient algorithm to calculate Whittle indices for the sub-MDP, inspired by the framework in \cite{Akbarzadeh_Mahajan_2022} for computing Whittle indices of MDPs with discounted rewards. 
    The computational complexity of the proposed algorithm is reduced compared to the referenced approach by leveraging the threshold structure and the derived analytical expressions of Whittle indices of the sub-MDP.
    Moreover, closed-form expressions for the Whittle index of the sub-MDP under error-free channels are also obtained based on the simplified sub-MDP structure. Simulation results validate our theoretical analysis and demonstrate the superior, asymptotically optimal performance of the proposed Whittle index policy compared to baseline schemes.

\end{itemize}

We note that Whittle index methods have been used to address AoI-related scheduling problems without considering the query arrival process (see, e.g., \cite{9241401,8437712,8935400,8919842}). 
In these works, the authors carefully analyzed the structures of the formulated RMABs and derived the value functions in closed form. They then established the Whittle indexability of these problems using the closed-form value functions.
However, in our problem, the query arrival processes at the BS result in bi-dimensional sub-MDP states and complex state transitions. This makes deriving closed-form value functions and subsequently a closed-form Whittle index challenging, if not impossible. Thus, the methods used to establish Whittle indexability in previous works are not applicable to our problem.


\textbf{\textit{Notations:}} In this paper, $\mathbb{Z}^+$ represents the set of positive integers, $\mathbb{E}[\cdot]$ denotes the operator of expectation, entity $[\cdot]$ is the representation of a vector containing the same type of elements, and $\langle\cdot \rangle$ denotes a tuple containing different types of elements.

\section{System Model and Problem Formulation}
\begin{figure}[t]
	\centering
	\includegraphics[width=0.45\textwidth]{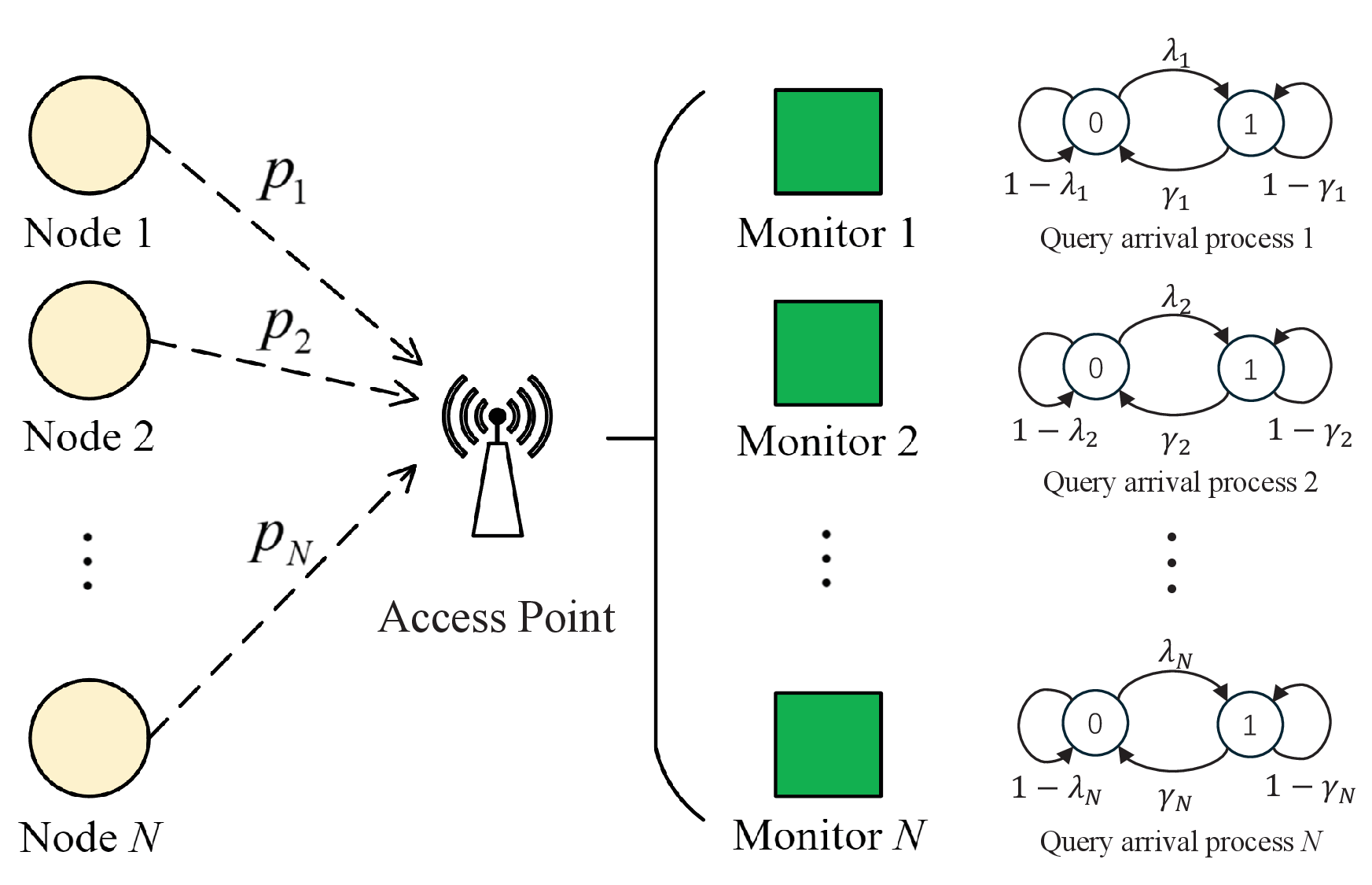}
	\caption{The wireless multiuser uplink network with Markovian query arrival processes at the monitors.}
	\label{Network}
\vspace{-1em} 
\end{figure}
\subsection{Network Model}
As depicted in Fig. \ref{Network}, we consider a multiuser wireless uplink
network where one base station (BS) coordinates and receives transmissions of time-sensitive status updates of $N$ end devices. 
Time is slotted and indexed by $t\in\{1,2,\cdots,T\}$, and the end devices, also called nodes hereafter, are indexed by $i\in\{1,2,\cdots,N\}$.
The status update packets of the nodes are delivered to the BS via $M$ ($0<M<N$) orthogonal error-prone channels, and one transmission takes one time slot.
A channel can be occupied by the status update transmission of only one node to avoid collision. 
Furthermore, a node can access at most one channel in each time slot. Assume that the status update packets of node $i$ are transmitted to the BS with success rate $p_i\in(0,1]$ in all $M$ channels.
The transmissions of status update packets in the uplink are centrally coordinated by the BS.
Further, the generate-at-will model is adopted for the generation of status updates.
Specifically, at the beginning of each time slot, the BS grants at most $M$ nodes, and the granted nodes generate their new status updates and transmit the status update packets to the BS concurrently. 
Let $a_i(t)$ denote the indicator whether node $i$ is scheduled in slot $t$, where $a_i(t)=1$, if node $i$ is scheduled, and $a_i(t)=0$, otherwise.
Then, we have $\sum_{i=1}^N a_i(t)\le M$.

We assume that each node has a corresponding monitor using the information of the status updates received from the node at the BS side.
The monitors do not always need the received status updates. Instead, each monitor $i$ uses the latest received status update from node $i$ only at the instants when queries arrive at monitor $i$.
The query arrival processes of nodes are characterized by independent binary Markov chains, as shown in Fig. \ref{Network}.
The state of query arrival process $i$ is denoted by $q_i(t)\in\{0,1\}$, where $q_i(t)=1$ represents that a query arrives at monitor $i$ in slot $t$, and $q_i(t)=0$ otherwise.
Moreover, $\lambda_i$ and $\gamma_i$ are transition probability from $q_i(t)=0$ to $q_i(t)=1$ and that in the opposite direction, respectively.

\subsection{Information Freshness Metric and Problem Formulation}\label{SectionAoI}
Since the information is considered by the monitors when the queries arrive, we employ the QAoI metric, originally proposed in \cite{9676636}, to capture the information freshness of all nodes at the BS.
To characterize the QAoI mathematically, we first introduce the AoI metric \cite{5984917}, defined as the time elapsed since the generation of the latest received message.
Accordingly, in the generate-at-will model, the evolution of the AoI of node $i$ in time slot $t$, denoted by $D_i(t)$, can be described by
\begin{equation}\label{EvoAoI}
	D_i(t+1)=\begin{cases}
		1, &\text{if status update of node}\ i\text{ is}\\
		&\text{received by the BS in slot}\ t,\\
		D_i(t) +1, &\text{otherwise}.
	\end{cases}
\end{equation}
Recalling the query mechanism of the network that $D_i(t)$ is useful for monitor $i$ only when $q_i(t)=1$, the definition of the QAoI of node $i$ in slot $t$ is given by $D_{q,i}(t)\triangleq D_i(t)q_i(t)$.

In this work, the long-term expected sum QAoI (ESQAoI) is embraced as the performance metric, which is mathematically defined as $\frac{1}{NT}\mathbb{E}\left[\sum^T_{t=1}\sum^N_{i=1}D_{q,i}(t)\Big|\bm{\pi}\right]$,
where $\bm{\pi}$ represents a stationary deterministic scheduling policy, and the expectation is computed over the entire system's dynamics.
The goal of this work is to devise a scheduling policy $\bm{\pi}$ for the BS to minimize the long-term ESQAoI. 
Specifically, we have the following optimization problem
\begin{equation}\label{P}
    \begin{split}
    &\min_{\bm{\pi}}\quad  \limsup_{T\to \infty}\frac{1}{NT}\mathbb{E}\left[\sum^T_{t=1}\sum^N_{i=1}D_{q,i}(t)\Big|\bm{\pi}\right],\\
    &\mbox{s.t.,}\quad 
    \sum^N_{i=1}a_i(t)\le M,\forall t.
    \end{split}
\end{equation}
The values of $\lambda_i$'s, $\gamma_i$'s, $p_i$'s, along with the local observations of AoI $D_i(t)$'s and query $q_i(t)$'s, are available to the BS for determining the scheduling of node groups at the beginning of each time slot.

\section{MDP Formulation and Relaxation}

\subsection{MDP Components}
The decision-making problem \eqref{P} can be naturally modeled as an MDP with the following components:

\underline{\textit{States:}} 
The state in slot $t$ can be denoted by $\bm{s}(t) \triangleq \langle \bm{D}(t),\bm{q}(t)\rangle$, where $\bm{D}(t)\triangleq \left[ D_1(t),\dots,D_N(t)\right] \in \{1,2,\cdots,D_m\}^N$ and $\bm{q}(t)\triangleq \left[ q_1(t),\dots,q_N(t)\right] \in \{0,1\}^N$, respectively.
Note that we impose an upper bound $D_m\in\mathbb{Z}^+$ on the AoI to ensure the state space of the MDP remains finite.

\underline{\textit{Actions:}} The action in slot $t$ is denoted by $\bm{a}(t)=[a_1(t),a_2(t),\cdots,a_N(t)]\in \{0,1\}^N$.
Since the scheduling constraint, we have $\sum^N_{i=1}a_i(t)\le M$.

\underline{\textit{Transition Functions:}} Denoted by $\Pr[\bm{s}(t+1)|\bm{s}(t),\bm{a}(t)]$, the transition function characterizes the probability of transiting from $\bm{s}(t)$ to $\bm{s}(t+1)$ when $\bm{a}(t)$ is taken.
As the evolutions of $\bm{D}(t)$ and $\bm{q}(t)$ are independent of each other, we have
\begin{equation}
\begin{split}
    &\Pr[\bm{s}(t+1)|\bm{s}(t),\bm{a}(t)]\\
    &=\Pr[\bm{D}(t+1)|\bm{D}(t),\bm{a}(t)]\Pr[\bm{q}(t+1)|\bm{q}(t)],
\end{split}
\end{equation}
where 
\begin{equation}
    \Pr\left[ \bm{D}(t+1)|\bm{D}(t),\bm{a}(t) \right]\!=\!\prod_{i=1}^{N}\Pr\left[D_i(t+1)|D_i(t),a_i(t) \right]
\end{equation}
with
\begin{equation}
\begin{split}  &\Pr[D_i(t+1)|D_i(t),a_i(t)]=\\
    &\begin{cases}
        p_i,& \text{if }a_i(t)=1, D_i(t+1)=1,\\
        1-p_i,& \text{if }a_i(t)=1, 
        D_i(t+1)=\min\{D_m,D_i(t)+1\},\\
        1,& \text{if }a_i(t)=0,
        D_i(t+1)=\min\{D_m,D_i(t)+1\},\\
        0,& \text{otherwises.}
    \end{cases}
\end{split}
\end{equation}
In addition, 
\begin{equation}
    \Pr\left[ \bm{q}(t+1)|\bm{q}(t) \right] =\prod_{i=1}^{N}\Pr\left[q_i(t+1)|q_i(t)\right],
\end{equation}
where
\begin{equation}\label{Trans_d}
	\Pr\left[ q_i(t+1)|q_i(t)\right]=
	\begin{cases}
		\lambda_i, &\text{if}\ q_i(t+1)=1, q_i(t)=0,\\
		1-\lambda_i, &\text{if}\ q_i(t+1)=0, q_i(t)=0,\\
        \gamma_i, &\text{if}\ q_i(t+1)=0, q_i(t)=1,\\
        1-\gamma_i, &\text{if}\ q_i(t+1)=1, q_i(t)=1,\\
		0, & \text{otherwise.}
	\end{cases} 
\end{equation}

\underline{\textit{Reward:}} The immediate reward in slot $t$ is defined as $r[\bm{s}(t)]\triangleq\sum_{i=1}^N D_{q,i}(t)$.
We consider the expected average reward $\frac{1}{NT}\mathbb{E}\{\sum_{t=1}^T r[\bm{s}(t)]|\pi\}$ for the MDP, which aligns with the long-term ESQAoI performance.
Another commonly applied reward is the expected discounted reward, given by $\mathbb{E}\{\sum_{t=1}^T \beta^{t-1}r[\bm{s}(t)]|\pi,\bm{s}(1)\}$ with an initial state~$\bm{s}(1)$, where $\beta\in(0,1)$ is the discounted factor.

We remark that the formulated MDP is equivalent to a restless multi-armed bandit (RMAB) with $N$ arms.
In particular, $a_i(t)=0$ corresponds to taking the passive action to arm $i$, and $a_i(t)=1$ means taking the active action to arm $i$ in slot~$t$.

The optimal policy for the considered scheduling problem can be found by solving the formulated MDP/RMAB through the standard dynamic programming (DP) method.
However, to ensure the accuracy of the optimal policy when applying the truncation, $D_m$ cannot be too small, thereby leading to a significantly large state space of the MDP when $N$ is large.
In light of this, tackling the formulated MDP via the DP method will suffer from the curse of dimensionality \cite{papadimitriou1987complexity}.
This drives us to develop a more efficient approach.

\subsection{MDP Relaxation for the Whittle Index approach}
We employ Whittle index method \cite{Whittle_1988}, which is a typical technique for solving RMABs.
This is because this method, which can decompose a multi-agent MDP into multiple sub-problems if the MDP is indexable, can be implemented with low complexity.
More specifically, Whittle index method assigns an index to each state of each sub-MDP and schedules the $M$ nodes with the highest indices in their current states at each time slot.
Additionally, Whittle index policy is known to be asymptotically optimal as the number of agents increases~\cite{Weber_Weiss_1990}.

Following Whittle index  approach \cite{Whittle_1988}, we first relax the scheduling constrain of problem \eqref{P}, yielding
\begin{equation}\label{PR}
    \begin{split}
    &\min_{\bm{\pi}}\quad  \limsup_{T\to \infty}\frac{1}{NT}\mathbb{E}\left[\sum^T_{t=1}\sum^N_{i=1}D_{q,i}(t)\Big|\bm{\pi}\right],\\
    &\mbox{s.t.,}\quad 
    \limsup_{T\to \infty}\frac{1}{T}\sum_{t=1}^T\sum^N_{i=1}a_i(t)\le M.
    \end{split}
\end{equation}
Remark that the minimization of problem \eqref{PR} is a lower bound to that of \eqref{P}.
Then, we can adopt Lagrange relaxation by introducing a dual variable $C\in \mathbb{R}$ and decouple the corresponding Lagrangian into $N$ sub-problems, given by
\begin{equation}\label{DPR}
    \min_{\pi_i}\quad  \limsup_{T\to \infty}\frac{1}{T}\sum^T_{t=1}\mathbb{E}\left[D_{q,i}(t)+Ca_i(t)\Big|\pi_i\right]
\end{equation}
for $i\in\{1,2,\cdots,N\}$, where $\pi_i$ is a stationary deterministic policy associated with node $i$.
Note that $C$ can be treated as the cost of scheduling. 
Clearly, sub-problem $i$ can be formulated as an independent sub-MDP $i$ with state $s_i=\langle D_i,q_i\rangle\in \mathcal{S}_i\triangleq \{1,2,\cdots,D_m\}\times \{0,1\}$, action $a_i\in \mathcal{A}_i\triangleq \{0,1\}$, and the expected average reward given by the cost function of \eqref{DPR}.
The optimal policy that achieves the minimum in \eqref{DPR} is denoted by $\pi^*_i$.
The mathematical definitions of these sub-MDPs are symmetric, thus we focus on an arbitrary sub-MDP $i$ in the following.
For brevity, we omit index $i$ hereafter unless it is necessary.

By Whittle index approach, an MDP can employ Whittle index policy if the corresponding sub-MDP is proven to be Whittle indexable, which depends on the structure of the corresponding sub-MDP.
To proceed, we investigate the formulated sub-MDP and present the following proposition:


\begin{proposition}\label{PropVF}
The formulated sub-MDP is an unichain MDP, which has the following properties.
Given a stationary deterministic policy $\pi$ for a sub-MDP, there exist a scalar $R^{\pi}$ and a function $h$ that satisfy the following Bellman’s equation
\begin{equation}
    R^{\pi}+h^{\pi}(s)=D^{q}+Ca^{\pi}+\sum_{s\in\mathcal{S}}\Pr(s'|s,a^{\pi})h^{\pi}(s'),
\end{equation}
where $R^{\pi}$ is the expected average reward under policy $\pi$ and satisfies
\begin{equation}
    R^{\pi}=\lim_{\beta\to 1}(1-\beta)V_{\beta}^{\pi}(s),\forall s\in\mathcal{S},
\end{equation}
and $h^{\pi}(s)$ is the relative cost function, presented as
\begin{equation}
    h^{\pi}(s)=\lim_{\beta\to 1}[V_{\beta}^{\pi}(s)-V_{\beta}^{\pi}(s_0)],\forall s\in\mathcal{S},
\end{equation}
for a reference state $s_0\in\mathcal{S}$, where $V_{\beta}^{\pi}(s)$ is the expected discounted cost function of the sub-MDP under $\pi$, given~by
\begin{equation}
    V_{\beta}^{\pi}(s)\!= \!\limsup_{T\to \infty}\sum^T_{t\!=\!1}\mathbb{E}\!\left[\beta^{t-1}\!(D_q(t)\!+\!Ca(t))\Big|\pi,\!s(1)\!=\!s\right].
\end{equation}
Moreover, there exists 
\begin{equation}
    \pi^*=\arg\min_{\pi}R^{\pi}=\arg\min_{\pi}R^{\pi}+h^{\pi}(s),
\end{equation}
\begin{equation}
    \pi^*_{\beta}=\arg\min_{\pi}V_{\beta}^{\pi}(s),
\end{equation}
for all $s\in\mathcal{S}$, such that $\lim_{\beta\to 1}\pi_{\beta}^*=\pi^*$.
\end{proposition}
\begin{proof}
    See Appendix \ref{Proof_PropVF}.  
\end{proof}

\begin{remark}
    $V_{\beta}^{\pi}(s)$ is the long-term reward of the formulated sub-MDP when we define the reward of the sub-MDP as the expected discounted reward $\mathbb{E}\{\sum_{t=1}^T \beta^{t-1}[D_q(t)+Ca(t)]|\pi,s(1)\}$ rather than the expected average reward.
    Proposition \ref{PropVF} reveals the relationship between the cost functions of the sub-MDPs under these two rewards and the relationship between the corresponding optimal policies.
    That is, the $R^{\pi}$ and $\pi^*$ are the limit points of $(1-\beta)V_{\beta}^{\pi}(s)$ and $\pi^*_{\beta}$, respectively, as $\beta\to 1$.
    Given that the formulated sub-MDP under the expected discounted reward is more tractable to analyze, we are motivated to comprehend the structure of the formulated sub-MDP by exploring $V_{\beta}^{\pi}(s)$ and $\pi^*_{\beta}$.
\end{remark}

Hereafter, we denote \textit{the formulated sub-MDP under the expected average reward} by $\mathcal{M}$, and \textit{the formulated sub-MDP under the expected discounted reward} by $\mathcal{M}_{\beta}$.
Furthermore, to facilitate the exploration of $V_{\beta}^{\pi}(s)$, we refine its expression~as
\begin{equation}
    V^{\pi}_{\beta}(s)=Q^{\pi}_{\beta}(s;a^{\pi})=J^{\pi}_{\beta}(s)+CA^{\pi}_{\beta}(s),
\end{equation}
where 
\begin{equation}
    J^{\pi}_{\beta}(s)=\limsup_{T\to \infty}\sum^T_{t\!=\!1}\mathbb{E}\!\left[\beta^{t-1}\!D_q(t)\Big|\pi,\!s(1)\!=\!s\right]
\end{equation}
denotes the expected discounted total QAoI, 
\begin{equation}
    A^{\pi}_{\beta}(s)=\limsup_{T\to \infty}\sum^T_{t\!=\!1}\mathbb{E}\!\left[\beta^{t-1}\!a(t)\Big|\pi,\!s(1)\!=\!s\right]
\end{equation}
represents the expected discounted total active times under policy $\pi$ and beginning at state $s$, and
\begin{equation}
\begin{split}
    &Q^{\pi}_{\beta}(s;a)\\
    &=\limsup_{T\to \infty}\sum^T_{t\!=\!1}\mathbb{E}\!\left[\beta^{t-1}\!(D_q(t)\!+\!Ca(t))\Big|\pi,\!s(1)\!=\!s,a(1)\!=\!a\right]\\
    &=J^{\pi}_{\beta}(s;a)+CA^{\pi}_{\beta}(s;a)
\end{split}
\end{equation}
is the action-cost function under policy $\pi$ and beginning at state $s$ and action $a$, where
\begin{equation}
    J^{\pi}_{\beta}(s;\!a)\!=\!\limsup_{T\to \infty}\!\sum^T_{t\!=\!1}\mathbb{E}\!\left[\beta^{t-1}\!D_q(t)\Big|\pi,\!s(1)\!=\!s,\!a(1)\!=\!a\right]\!,
\end{equation}
and
\begin{equation}
    A^{\pi}_{\beta}(s;\!a)\!=\!\limsup_{T\to \infty}\sum^T_{t\!=\!1}\mathbb{E}\!\left[\beta^{t-1}\!a(t)\Big|\pi,\!s(1)\!=\!s,a(1)\!=\!a\right].
\end{equation}
We denote $R^{\pi}$, $V^{\pi}_{\beta}(s)$, $Q^{\pi}_{\beta}(s;a)$, $J^{\pi}_{\beta}(s)$, and $A^{\pi}_{\beta}(s)$ under the corresponding optimal policies by $R^*$, $V^*_{\beta}(s)$, $Q^*_{\beta}(s;a)$, $J^*_{\beta}(s)$, and $A^*_{\beta}(s)$, respectively.
Besides, we can write
\begin{equation}
\begin{split}
    V^*_{\beta}(s)=\min\{Q^*_{\beta}(s;a=0),Q^*_{\beta}(s;a=1)\}.
\end{split}  
\end{equation}
On this basis, we verify the indexability of the formulated sub-MDP in the next section.

\section{Whittle Indexability}

In this section, we analyze the structure of $\mathcal{M}_{\beta}$ for establishing Whittle indexability of the formulated sub-MDP.
To proceed, we first introduce a definition of the indexability, stated as follows.

\begin{definition}\label{DWA}
    Given a sub-MDP under the expected average reward, we define the passive set as $\mathcal{P}(C)\triangleq\{s\in\mathcal{S}|\pi^*(s)=0\}$, which represents the subset of all states that will take passive action under at least one optimal policy.
    Then, the sub-MDP is Whittle indexable if the following properties hold:
    \begin{enumerate}
        \item[a)] $\mathcal{P}(-\infty)=\emptyset$, and $\mathcal{P}(+\infty)=\mathcal{S}$.
        \item[b)] $\mathcal{P}(C_1)\subseteq\mathcal{P}(C_2)$ for $C_1
        \le C_2$.
    \end{enumerate}
\end{definition}
Then, if Definition \ref{DWA} is satisfied, the Whittle index is defined as
\begin{equation}\label{TDWA}
    W(s)=\inf\{C|s\in\mathcal{P}(C)\},
\end{equation}
equivalently the scheduling cost $C$ that makes the active and passive actions equally desirable for $\pi^*$.
Similarly, we have
\begin{definition}\label{DWD}
    Given a sub-MDP under the expected discounted reward, we define the passive set as $\mathcal{P}_{\beta}(C)\triangleq\{s\in\mathcal{S}|\pi^*_{\beta}(s)=0\}$.
    Then, the sub-MDP is Whittle indexable if the following properties are satisfied:
    \begin{enumerate}
        \item[a)] $\mathcal{P}_{\beta}(-\infty)=\emptyset$, and $\mathcal{P}_{\beta}(+\infty)=\mathcal{S}$.
        \item[b)] $\mathcal{P}_{\beta}(C_1)\subseteq\mathcal{P}_{\beta}(C_2)$, for $C_1
        \le C_2$.
    \end{enumerate}
\end{definition}
Then, if Definition \ref{DWD} holds, the Whittle index is defined as
\begin{equation}
    W_{\beta}(s)=\inf\{C|s\in\mathcal{P}_{\beta}(C)\},
\end{equation}
i.e., the cost $C$ that makes $Q^*_{\beta}(s,a=0)=Q^*_{\beta}(s,a=1)$.

Based on Proposition \ref{PropVF} implying that $\pi^*$ has the same properties as $\pi^*_{\beta}$, we aim to verify whether Definition \ref{DWA} holds for $\mathcal{M}$ by verifying whether Definition \ref{DWD} holds for $\mathcal{M}_{\beta}$.

To proceed, we explore the structure of $V^*_{\beta}(s)$ and $\pi^*_{\beta}$.
According to \cite{sennott1989average}, $V^*_{\beta}(s)$ can be obtained by iteration method.
Specifically, let $V_{\beta,n}(s)$ ($n\in\mathbb{Z}^+$) be the cost-to-go function such that $V_{\beta,1}(s)=0$ for all $s\in\mathcal{S}$, and $\bar{D}\triangleq\min\{D+1,D_m\}$.
Then, we have
\begin{equation}
    V_{n+1,\beta}(D,\!q)\!=\!\min\!\left\{Q_{n+1,\beta}(D,\!q;0),\!Q_{n+1,\beta}(D,q;1)\right\}\!,
\end{equation}
where
\begin{equation}\label{Q00}
\begin{split}
    &Q_{n+1,\beta}(D,0;0)\\
    &=\beta\left[(1-\lambda)V_{n,\beta}(\bar{D},0)+\lambda V_{n,\beta}(\bar{D},1)\right],
\end{split}
\end{equation}
\begin{equation}
\begin{split}
    &Q_{n+1,\beta}(D,1;0)\\
    &=D+\beta\left[\gamma V_{n,\beta}(\bar{D},0)+(1-\gamma) V_{n,\beta}(\bar{D},1)\right],
\end{split}
\end{equation}
and $Q_{n+1,\beta}(D,0;a=1)$, $Q_{n+1}(D,1;a=1)$, shown in \eqref{Q01} and \eqref{Q11} at the top of the next page.
\begin{figure*}[ht]
\centering
\begin{equation}\label{Q01}
\begin{split}
    Q_{n+1,\beta}(D,0;1)
    =C\!+\!\beta\left[(1\!-\!p)(1\!-\!\lambda) V_{n,\beta}(\bar{D},0)
    +(1\!-\!p)\lambda V_{n,\beta}(\bar{D},1)
    +p(1\!-\!\lambda)V_{n,\beta}(1,0)+p\lambda V_{n,\beta}(1,1)\right].
\end{split}  
\end{equation}
\begin{equation}\label{Q11}
\begin{split}
    Q_{n+1,\beta}(D,1;1)=C\!+\!D\!+\!\beta\!\left[(1\!-\!p)\gamma V_{n,\beta}(\bar{D},0)+(1\!-\!p)(1\!-\!\gamma) V_{n,\beta}(\bar{D},1)
    +p\gamma V_{n,\beta}(1,0)+p(1\!-\!\gamma) V_{n,\beta}(1,1)\right].
\end{split}  
\end{equation} 
\vspace{-0.5cm}
\hrulefill
\end{figure*}

Then, $V_{n,\beta}(s)\to V^*_{\beta}(s)$ as $n\to\infty$ for all $s\in\mathcal{S}$ and $\beta$, and we can propound the following lemma.

\begin{lemma}\label{Vp}
    $V^*_{\beta}(D,q)$ is non-decreasing with regard to $D$.
\end{lemma}
\begin{proof}
    See Appendix \ref{Proof_Vp}.
\end{proof}

Built upon Lemma \ref{Vp}, we can identify the threshold structure of the optimal policy $\pi^*_{\beta}$, as presented in the following theorem.
\begin{theorem}\label{ThP}
    Given $\lambda,\gamma\in(0,1),\beta\in(0,1)$, and $C\ge 0$, the optimal policy $\pi^*_{\beta}$ of $\mathcal{M}_{\beta}$ has a threshold structure in $D$, i.e., given $q$, there exists a threshold $H^*_q(C)\in\{1,2,\cdots, D_m+1\}$ such that it is optimal to schedule when $D\ge H^*_q(C)$ and optimal to remain idle when $D<H^*_q(C)$.
\end{theorem}
\begin{proof}
    See Appendix \ref{Proof_ThP}.
\end{proof}

Define $\pi(H_0;H_1)$ as a threshold-type policy which chooses to be idle when $D<H_0$ for $q=0$ and $D<H_1$ for $q=1$, while chooses to schedule when $D\ge H_0$ for $q=0$ and $D\ge H_1$ for $q=1$.
Let $\Pi_h$ denote the set of all possible $\pi(H_0;H_1)$, i.e., $\Pi_h=\{\pi(H_0;H_1)|1\le H_q \le D_m+1\}$.
Recall that $\pi^*$ is the limit point of $\pi^*_{\beta}$ as $\beta\to 1$.
Obviously, $\pi^*,\pi^*_{\beta}\in\Pi_h$.
Based on Theorem \ref{ThP}, we have the following corollary.
\begin{corollary}\label{EqualTh}
    Given $C$, if $\lambda+\gamma=1$, the optimal threshold-type policy $\pi^*_{\beta}$ of $\mathcal{M}_{\beta}$ has $H^*_0(C)=H^*_1(C)$.
\end{corollary}
\begin{proof}
    See Appendix \ref{Proof_EqualTh}.
\end{proof}

Conventionally, further theoretical analysis is required to obtain the closed-form cost function $V^*_{\beta}(s)$'s of $\mathcal{M}_{\beta}$, which can lay the foundation for verifying Whittle indexability, i.e., Definition \ref{DWD}. 
However, it is well-known that deriving closed-form expressions for $V^*_{\beta}(s)$ is challenging, especially when the state space is multi-dimensional. For our problem, the binary Markovian query
arrival process results in a bi-dimensional state and complex
state transitions within each sub-MDP. 
To circumvent this problem, we apply the active time (AT) condition, recently introduced in \cite[Theorem 3]{zhoueasier}, which offers a means to simplify the verification of Whittle indexability.
The condition is presented as follows.
\begin{theorem}\label{TAT}
    (The AT condition) A sub-MDP is indexable if the following conditions hold:
    \begin{enumerate}
        \item[a)] $\mathcal{P}_{\beta}(-\infty)=\emptyset$, and $\mathcal{P}_{\beta}(+\infty)=\mathcal{S}$. 
        \item[b)] $A^{\pi}_{\beta}(s;1)\ge A^{\pi}_{\beta}(s;0),\forall s\in\mathcal{S}$ and $\forall \pi\in\Pi'$, where $\Pi'$ denotes the collection of all policies that are possible to be optimal. 
    \end{enumerate}
\end{theorem}
Remark that by Theorem \ref{ThP}, we can let $\Pi_h$ be $\Pi'$ when the AT condition is applied to verify the indexability of the considered sub-MDP.
In addition, built on Theorem \ref{ThP}, we can obtain the following proposition.
\begin{proposition}\label{VAT}
    For $\mathcal{M}_{\beta}$, we have $A^{\pi}_{\beta}(s;1)\ge A^{\pi}_{\beta}(s;0),\forall s\in\mathcal{S}$ and $\forall \pi\in\Pi_h$.
\end{proposition}
\begin{proof}
    See Appendix \ref{Proof_VAT}.
\end{proof}
It is straightforward that $\pi^*_{\beta}$ will take $a=1$ for all states when $C\le 0$, and will take $a=0$ for all states when $C=+\infty$, i.e., condition a) of Theorem \ref{TAT} holds for $\mathcal{M}_{\beta}$.
Hence, we can conclude that $\mathcal{M}_{\beta}$ satisfies the AT condition, and thus is Whittle indexable.
That is, Definition \ref{DWD} holds for $\pi^*_{\beta}$.
Based on Proposition \ref{PropVF}, $\pi^*$ is also threshold-type, and satisfies Definition \ref{DWA}, i.e., $\mathcal{M}$ is also Whittle indexable, since these properties are independent of $\beta$.
In light of this, we aim to design effective algorithms for calculating Whittle indices of $\mathcal{M}$, corresponding to the considered problem \eqref{P}.

\section{Algorithms to calculate Whittle Index}

In this section, 
we first develop an algorithm inspired by \cite[Algorithm 1]{Akbarzadeh_Mahajan_2022} to efficiently calculate Whittle indices for the formulated MDP under the expected average reward. Subsequently, we devise a lower-complexity algorithm for error-free scenarios.

\subsection{Algorithm for Error-Prone Networks}
Ref. \cite[Algorithm 1]{Akbarzadeh_Mahajan_2022} proposed a universal algorithm that implements Whittle index calculation for any Whittle indexable MDPs.
However, \cite[Algorithm 1]{Akbarzadeh_Mahajan_2022} is applied to the MDPs under the expected discounted reward.
As such, we can only approximate Whittle indices of the considered problem with this algorithm by converging $\beta$ to $1$.
By a more in-depth analysis of $\mathcal{M}$, we develop a more efficient algorithm, which can produce more accurate Whittle indices for the considered problem, based on \cite[Algorithm 1]{Akbarzadeh_Mahajan_2022}. 

To proceed, we make the following definitions for $\mathcal{M}$.
Let $N_W$ represent the number of distinct Whittle indices.
Note that $N_W\le |\mathcal{S}|$.
Let $\Theta=\{C_1,C_2,\cdots,C_{N_W}\}$ where $C_1<C_2<\cdots<C_{N_W}$ denote the sorted list of distinct Whittle indices from smallest to largest.
Moreover, we define $C_0=0$.
Let $\mathcal{R}_w=\{\langle D,q\rangle\in\mathcal{S}:W(D,q)\le C_w\}$ represent the set of states whose corresponding Whittle index is less than or equal to $C_w$, for $w\in\{0,1,2,\cdots,N_W\}$.
Clearly, we have $\mathcal{R}_0=\emptyset$ and $\mathcal{R}_{N_W}=\mathcal{S}$.
Also, for $w\in\{0,1,2,\cdots,N_W-1\}$, let $\Gamma_{w+1}\gets\mathcal{R}_{w+1}\backslash \mathcal{R}_w$ denote the set of the states that have Whittle index $C_{w+1}$.

For any subset $\mathcal{X}\subseteq \mathcal{S}$, define a policy $\pi^\mathcal{X}$ such that
\begin{equation}
    \pi^\mathcal{X}(s)=
    \begin{cases}
        0, & \text{if }s\in \mathcal{X},\\
        1, & \text{otherwise.}
    \end{cases}
\end{equation}

Furthermore, we define
\begin{equation}\label{zeta}
    \zeta[\pi_1,\pi_2]\triangleq\frac{J^{\pi_1}-J^{\pi_2}}{A^{\pi_2}-A^{\pi_1}},
\end{equation}
where
\begin{equation}
    J^{\pi}\triangleq \lim_{\beta\to 1}(1-\beta)J^{\pi}_{\beta}(s),
\end{equation}
\begin{equation}
    A^{\pi}\triangleq \lim_{\beta\to 1}(1-\beta)A^{\pi}_{\beta}(s),
\end{equation}
for all $s\in\mathcal{S}$, are the expected average QAoI and the expected average active times of $\mathcal{M}$ under policy $\pi$.
Then, we have the following theorem.
\begin{theorem}\label{CalWI}
    For $\mathcal{M}$, we have
    \begin{enumerate}
        \item $\pi^{\mathcal{R}_w}\in \Pi_h,\forall w\in\{0,1,2,\cdots,N_W\}$.
        \item Assume $\pi^{\mathcal{R}_w}=\pi(H_0;H_1)$.
        Then, there exists $\bar{H}_0$ and $\bar{H}_1$ satisfying $H_0\le \bar{H}_0\le D_m$, $H_1\le \bar{H}_1\le D_m$, and $\bar{H}_0+\bar{H}_1-H_0-H_1>0$, such that $\pi^{\mathcal{R}_{w+1}}=\pi(\bar{H}_0;\bar{H}_1)$.
        This indicates
        \begin{equation}\label{Cse1}
        \begin{split}
            \Gamma_{w+1}&=\{\langle H_0+1,0 \rangle,\langle H_0+2,0 \rangle,\cdots,\langle \bar{H}_0,0 \rangle,\\
            &\langle H_1+1,1 \rangle,\langle H_1+2,1 \rangle,\cdots,\langle \bar{H}_1,1 \rangle\}.
        \end{split}
        \end{equation} 
        \item The value of Whittle index $C_{w+1}$ can be obtained by
        \begin{equation}\label{Cnexteq}
        \begin{split}
            C_{w+1}&=\zeta[\pi(H_a;H_1),\pi(H_0;H_1)]\\
            &=\zeta[\pi(H_0;H_b),\pi(H_0;H_1)],
        \end{split}
        \end{equation}
        for all $H_a\in\{H_0+1,\cdots,\bar{H}_0\}$ and $H_b\in\{H_1+1,\cdots,\bar{H}_1\}$.
        Note that $H_a$ or $H_b$ do not exist if $H_0=\bar{H}_0$ or $H_1=\bar{H}_1$.
        
        In addition, for all $s\in \Gamma_{w+1}$, $W(s)=C_{w+1}$.
    \end{enumerate}
    
\end{theorem}
\begin{proof}
    See Appendix \ref{Proof_CalWI}.
\end{proof}

As stated in Theorem \ref{CalWI}, to determine Whittle indices for the sub-MDPs, we need to compute $J^{\pi}$ and $A^{\pi}$ for all $\pi \in \Pi_h$.
\begin{figure*}[t]
	\centering
	\includegraphics[width=0.95\textwidth]{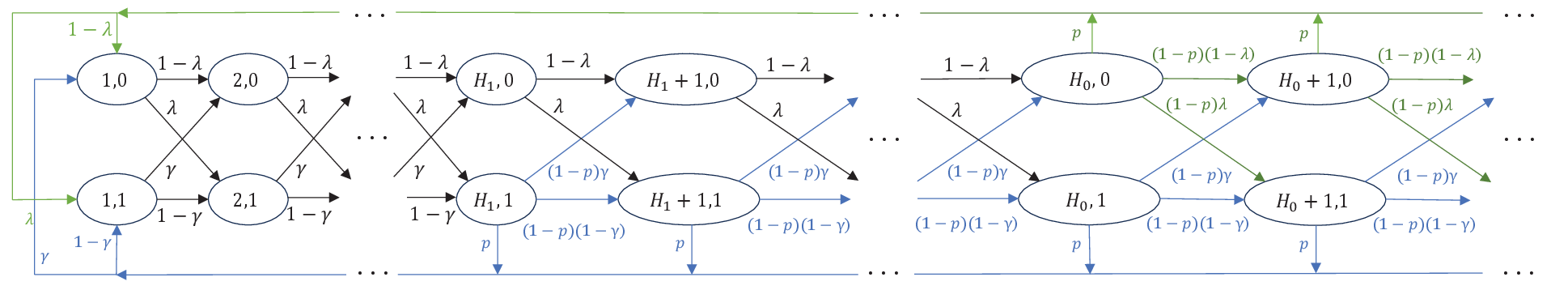}
	\caption{An example of a sub-MDP applying a thresh-type policy $\pi(H_0;H_1)$ with $H_1<H_0<D_m$, $p<1$.}
 \label{MCTH}
\end{figure*}
To this end, we consider the steady-state distribution of the sub-MDP applying $\pi \in \Pi_h$ when the channel is error-prone.
As depicted in Fig. \ref{MCTH}, we denote the steady-state probability of state $\langle D,q\rangle$ by $\mu_{D,q}$, and define $\bm{\mu}_D\triangleq [\mu_{D,0},\mu_{D,1}]^T$.
We state that all $\mu_{D,q}$ can be characterized by $\mu_{1,0}$ and $\mu_{1,1}$.
More specifically, let $H=\min\{H_0,H_1\}$ and $\hat{H}=\max\{H_0,H_1\}$.
Based on Fig. \ref{MCTH}, we can define $\mathbf{P}_1$, $\mathbf{P}_2$, $\mathbf{P}_3$, and $\mathbf{P}_4$, such that $\bm{\mu}_D=\mathbf{P}_1\bm{\mu}_{D-1}$ for $2\le D \le H$, $\bm{\mu}_D=\mathbf{P}_2\bm{\mu}_{D-1}$ for $H+1\le D \le \hat{H}$, $\bm{\mu}_D=\mathbf{P}_3\bm{\mu}_{D-1}$ for $\hat{H}+1\le D < D_m$, and $\bm{\mu}_{D_m}=\mathbf{P}_4\bm{\mu}_{D_m-1}$ for $D = D_m$.
To move forward, we introduce the following proposition on the expressions of these matrices.
\begin{proposition}\label{TransMatrix}
    Given $\mathcal{M}$ applying a policy $\pi(H_0;H_1)\in\Pi_h$, we have
    \begin{equation}\label{P1}
    \mathbf{P}_1=
    \begin{bmatrix}
        1-\lambda & \gamma\\
        \lambda & 1-\gamma
    \end{bmatrix},
\end{equation}
\begin{equation}\label{P2}
    \mathbf{P}_2=
    \begin{cases}
        \begin{bmatrix}
            1-\lambda & (1-p)\gamma\\
            \lambda & (1-p)(1-\gamma)
        \end{bmatrix},& \text{if }H_0>H_1,\\
        \begin{bmatrix}
            (1-p)(1-\lambda) & \gamma\\
            (1-p)\lambda & 1-\gamma
        \end{bmatrix},& \text{if }H_0<H_1,\\
        \mathbf{I},& \text{otherwise},
    \end{cases}
\end{equation}
\begin{equation}\label{P3}
    \mathbf{P}_3=
    \begin{bmatrix}
        (1-p)(1-\lambda) & (1-p)\gamma\\
        (1-p)\lambda & (1-p)(1-\gamma)
    \end{bmatrix}=
    (1-p)\mathbf{P}_1,
\end{equation}
and
\begin{equation}\label{P4}
    \mathbf{P}_4=
    \begin{cases}
        (\mathbf{I}-\mathbf{P}_3)^{-1}\mathbf{P}_3 & \text{if }H\le\hat{H}<D_m,\\
        (\mathbf{I}-\mathbf{P}_3)^{-1}\mathbf{P}_2 & \text{if }H<\hat{H}=D_m,\\
        (\mathbf{I}-\mathbf{P}_3)^{-1}\mathbf{P}_1 & \text{if }\hat{H}=H=D_m,\\
        (\mathbf{I}-\mathbf{P}_2)^{-1}\mathbf{P}_2 & \text{if }H<D_m,\hat{H}=D_m+1,\\
        (\mathbf{I}-\mathbf{P}_2)^{-1}\mathbf{P}_1 & \text{if }H=D_m,\hat{H}=D_m+1,\\
        (\mathbf{I}-\mathbf{P}_1)^{-1}\mathbf{P}_1 & \text{if }H=\hat{H}=D_m+1.\\
    \end{cases}
\end{equation}

\end{proposition}

\begin{proof}
    See Appendix \ref{Proof_TransMatrix}.
\end{proof}
Clearly, the characterization of $\mu_{D,q}$ in terms of $\mu_{D,0}$ and $\mu_{D,1}$ requires calculating multiple powers of $\mathbf{P}_1$, $\mathbf{P}_2$, and $\mathbf{P}_3$.
Accordingly, we have
\begin{equation}\label{CpP1}
    \mathbf{P}^n_1=(\mathbf{P}^T)^n=
    \begin{bmatrix}
    \frac{\lambda(1-\lambda-\gamma)^n+\gamma}{\lambda+\gamma} & \frac{-\gamma(1-\lambda-\gamma)^n+\gamma}{\lambda+\gamma}\\
    \frac{-\lambda(1-\lambda-\gamma)^n+\lambda}{\lambda+\gamma} & \frac{\gamma(1-\lambda-\gamma)^n+\lambda}{\lambda+\gamma}
    \end{bmatrix}
\end{equation}
and
\begin{equation}\label{CpP3}
    \mathbf{P}^n_3=(1-p)^n\mathbf{P}^n_1.
\end{equation}
The only numerical calculation necessary is $\mathbf{P}^n_2$.

Based on the state transitions shown in Fig. \ref{MCTH}, we can present
\begin{equation}\label{LP1}
    \sum_{D=1}^{D_m}(\mu_{D,0}+\mu_{D,1})=1,
\end{equation}
\begin{equation}\label{LP2}
    \mu_{1,0}=(1-\lambda)p\sum_{D=H_0}^{D_m}\mu_{D,0}+\gamma p\sum_{D=H_1}^{D_m}\mu_{D,1},
\end{equation}
\begin{equation}\label{CJ}
    J^{\pi(H_0;H_1)}=\sum_{D=1}^{D_m}D\mu_{D,1},
\end{equation}
and
\begin{equation}\label{CA}
    A^{\pi(H_0;H_1)}=\sum_{D=H_0}^{D_m}\mu_{D,0}+\sum_{D=H_1}^{D_m}\mu_{D,1}.
\end{equation}

Characterizing $\mu_{D,q}$ by expressions of $\mathbf{P}_1$-$\mathbf{P}_4$ makes \eqref{LP1} and \eqref{LP2} a system of linear equations in two variables $\mu_{1,0}$ and $\mu_{1,1}$.
Solving the linear system and substituting the solutions $\mu_{1,0}$ and $\mu_{1,1}$ into \eqref{CJ} and \eqref{CA}, we can obtain the values for the expected average QAoI and active time.
The more detailed expressions of Eqs. \eqref{LP1}-\eqref{CA} in different cases are given in Appendix \ref{Expression_Error_Channel}.

In view of Theorem \ref{CalWI} and Eqs. \eqref{LP1}-\eqref{CA}, Whittle indices of all states of a sub-MDP can be calculated.
We present the specific algorithm in Algorithm \ref{AlgErrorProne}, where 
\begin{equation}
    \bar{H}_0\triangleq \min\{H\!\in\!(H_0,D_m+1]|A^{\pi(H;H_1)}\neq A^{\pi(H_0;H_1)}\},
\end{equation}
\begin{equation}
    \bar{H}_1\triangleq \min\{H\!\in\!(H_1,D_m+1]|A^{\pi(H_0;H)}\neq A^{\pi(H_0;H_1)}\},
\end{equation}
and $\zeta[\pi_1,\pi_2]$ is calculated by following \eqref{LP1}, \eqref{LP2}, \eqref{CJ}, \eqref{CA}, and \eqref{zeta}.

\begin{algorithm}[t]
    \caption{Efficiently Calculating Whittle indices of all states of the considered sub-MDP with $p<1$}\label{AlgErrorProne}
    \textbf{Initialization}: $D_m$, $\lambda$, $\gamma$, $p$, $w\gets 0$, $\mathcal{R}_w\gets \emptyset$, and $H_0\gets 1$, $H_1\gets 1$\;
    \While{$\mathcal{R}_w\neq \mathcal{S}$}{
    \uIf{$H_1<H_0=D_m+1$}{
    $C_{w+1}\gets \zeta[\pi(H_0;\bar{H}_1),\pi(H_0;H_1)]$\;
    $\Gamma_{w+1}\gets\{\langle H_1+1,1 \rangle,\cdots,\langle\bar{H}_1,1\rangle\}$\;
    $H_1\gets \bar{H}_1$\;
    }
    \uElseIf{$H_0<H_1=D_m+1$}{
    $C_{w+1}\gets \zeta[\pi(\bar{H}_0;H_1),\pi(H_0;H_1)]$\;
    $\Gamma_{w+1}\!\gets\!\{\langle H_0\!+\!1,0 \rangle,\cdots,\langle\bar{H}_0,0\rangle\}$\;
    $H_0\gets \bar{H}_0$\;
    }
    \Else{
        \uIf{$\lambda+\gamma=1$}{
        $C_{w+1}\gets \zeta[\pi(\bar{H}_0;\bar{H}_1),\pi(H_0;H_1)]$\;
        $\Gamma_{w+1}\!\gets\!\{\langle H_0\!+\!1,0 \rangle,\langle H_1\!+\!1,1 \rangle,\cdots,\langle\bar{H}_0,0\rangle,\langle\bar{H}_1,1\rangle\}$\;
            $H_0\gets \bar{H}_0,H_1\gets \bar{H}_1$\;
        }
        \Else{
            $C_{w+1}\gets \min\{\zeta[\pi(\bar{H}_0;H_1),\pi(H_0;H_1)]$, $\zeta[\pi(H_0;\bar{H}_1);\pi(H_0;H_1)]\}$\;
            \uIf{
            $\zeta[\pi(\bar{H}_0;H_1),\pi(H_0;H_1)]\!<\!\zeta[\pi(H_0;\bar{H}_1),\pi(H_0;H_1)]$
            }
            {$\Gamma_{w+1}\gets\{\langle H_0+1,0 \rangle,\cdots,\langle\bar{H}_0,0\rangle\}$\;
            $H_0\gets \bar{H}_0$\;}
            \uElseIf{
            $\zeta[\pi(\bar{H}_0;H_1),\pi(H_0;H_1)]\!>\!\zeta[\pi(H_0;\bar{H}_1),\pi(H_0;H_1)]$
            }
            {
            $\Gamma_{w+1}\gets\{\langle H_1+1,1 \rangle,\cdots,\langle\bar{H}_1,1\rangle\}$\;
            $H_1\gets \bar{H}_1$\;
            }
            \Else{
            $\Gamma_{w+1}\!\gets\!\{\langle H_0\!+\!1,0 \rangle,\langle H_1\!+\!1,1 \rangle,\cdots,\langle\bar{H}_0,0\rangle,\langle\bar{H}_1,1\rangle\}$\;
            $H_0\gets \bar{H}_0,H_1\gets \bar{H}_1$\;
            }
        }
    }
    $W(s)\gets C_{w+1},\forall s\in \Gamma_{w+1}$\;
    $\mathcal{R}_{w+1}\gets \mathcal{R}_{w}\bigcup \Gamma_{w+1}$\;
    $w\gets w+1$\;
    }
\end{algorithm}

Following Eqs. \eqref{LP1}-\eqref{CA} and the structure of Algorithm \ref{AlgErrorProne}, the main proportion of computational complexity of Algorithm \ref{AlgErrorProne} lie in the calculations of $\mathbf{P}_d,\mathbf{P}^2_d,\cdots,\mathbf{P}^{D_m-2}_d$, for $d=1,2,3$, $\mathbf{P}_4$, $\bm{\mu}_1$, $J^{\pi}$, and $A^{\pi}$.
In particular, since the closed-from expressions given in \eqref{CpP1} and \eqref{CpP3}, calculating $\mathbf{P}^n_1$ and $\mathbf{P}^n_3$ requires $\mathcal{O}(\log_2n)$ calculations by the fast power method, calculating $\mathbf{P}^n_2$ requires $\mathcal{O}(8\log_2n)$ calculations by the matrix fast power method, and calculating $\mathbf{P}_4$ needs $\mathcal{O}(6\cdot 2^3)$.
Moreover, for each update of $H_0$ or $H_1$, Algorithm \ref{AlgErrorProne} requires $\mathcal{O}(2^3)$ calculations to calculate $\bm{\mu}_1$ by solving the linear equations \eqref{LP1} and \eqref{LP2}, and $\mathcal{O}(3D_m-H_0-H_1-2)$ to calculate $J^{\pi(H_0;H_1)}$ and $A^{\pi(H_0;H_1)}$.
Therefore, the computational complexity of Algorithm \ref{AlgErrorProne} is
\begin{equation}
\begin{split}
    &\mathcal{O}(2\log_2(D_m\!-\!1))+\mathcal{O}(8\log_2(D_m\!-\!1))+\mathcal{O}(6\cdot 2^3)\\
    &+(2D_m+2)\mathcal{O}(3D_m-H_0-H_1+2^3)=\mathcal{O}(D^2_m),
\end{split}
\end{equation}
which is more efficient than the complexity $\mathcal{O}(D_m^3)$ of \cite[Algorithm 1]{Akbarzadeh_Mahajan_2022} implemented in a matrix calculation method.

\subsection{Simplified Algorithm for Error-Free Networks}
\begin{figure*}[t]
	\centering
	\includegraphics[width=0.85\textwidth]{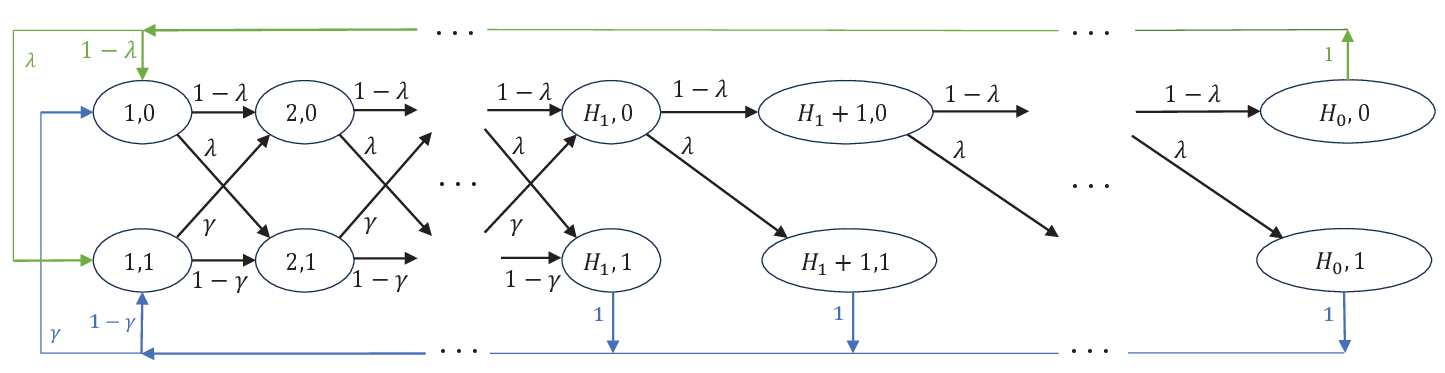}
	\caption{An example of a sub-MDP applying a thresh-type policy $\pi(H_0;H_1)$ with $H_1<H_0\le D_m$, $p=1$.}
 \label{MCTHnoError}
\end{figure*}

We now consider the error-free scenario as a special case. {The structure of a sub-MDP with an error-free channel (i.e., $p=1$), as illustrated in Fig. \ref{MCTHnoError},} is significantly simpler compared to the sub-MDP with an error-prone channel shown in Fig. \ref{MCTH}.

Specifically, in this context, $\mathbf{P}_3$ does not exist, and
\begin{equation}\label{P2noError}
    \mathbf{P}^n_2=
    \begin{cases}
        \begin{bmatrix}
            (1-\lambda)^n & 0\\
            \lambda(1-\lambda)^{n-1} & 0
        \end{bmatrix},& \text{if }H_0>H_1,\\
        \begin{bmatrix}
            0 & \gamma(1-\gamma)^{n-1}\\
            0 & (1-\gamma)^n
        \end{bmatrix},& \text{if }H_0<H_1,\\
        \mathbf{I},& \text{otherwise}.
    \end{cases}
\end{equation}
Furthermore, the states $\langle D,q \rangle$ are infeasible for $\hat{H}<D\le D_m$ due to $p=1$.
Consequently, \eqref{LP1}, \eqref{LP2}, \eqref{CJ}, and \eqref{CA} can be characterized by $\mu_{1,0}$ and $\mu_{1,1}$ in closed form.
The closed-form expressions of Eqs. \eqref{LP1}-\eqref{CA} in terms of $\mu_{1,0}$ and $\mu_{1,1}$ in different contexts are given in Appendix \ref{Expression_No_Error_Channel}.

In this sense, we can enhance the calculation efficiency of Algorithm \ref{AlgErrorProne} in this scenario by replacing the adopted analytical expressions of linear equations, $J^{\pi(H_0;H_1)}$ and $A^{\pi(H_0;H_1)}$ with the derived closed-form counterparts, and the initialization $p=1$.
In particular, the computational complexity of the enhanced algorithm in this context is
\begin{equation}
\begin{split}
(2D_m+2)\mathcal{O}(\log_2D_m+2^3)=\mathcal{O}(D_m\log_2D_m).
\end{split}
\end{equation}

\section{Numerical Results}
In this section, we first verify our theoretical analysis via simulation results.
Subsequently, the performance of Whittle index policy obtained by the proposed algorithm is assessed in comparison to two baseline policies: Greedy policy and Whittle index policy obtained by the approach introduced in the literature.  
We obtain each optimal policy by the standard DP approach, each group of Whittle indices by the proposed algorithms, and each performance curve in the following figure by averaging over 100000 Monte-Carlo simulation runs.

\subsection{Verification of the Analysis}

\begin{figure}[t]
	\centering
	\includegraphics[width=0.48\textwidth]{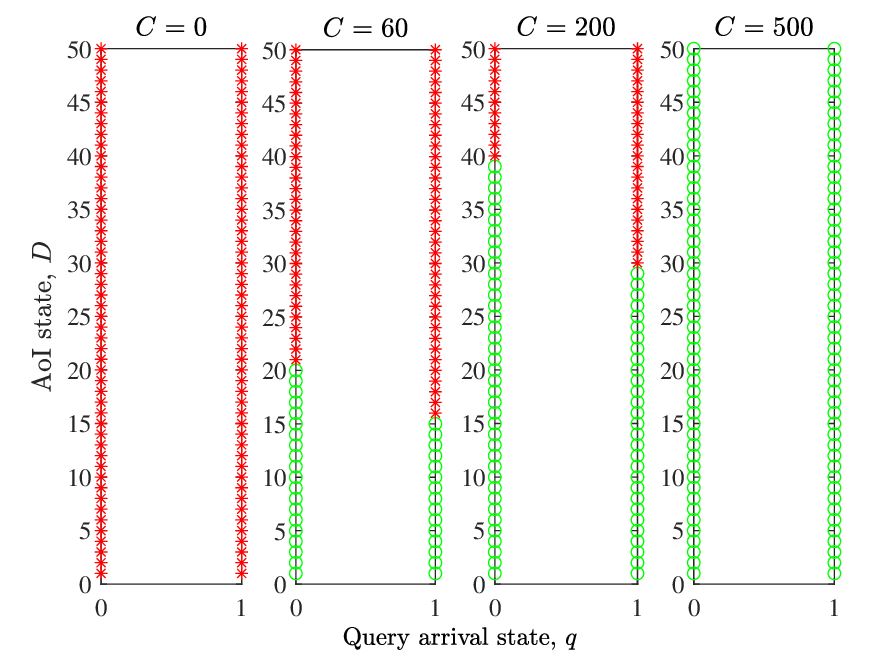}
	\caption{The optimal policy of a sub-MDP versus the scheduling cost $C$, where $\lambda=0.4$, $\gamma=0.3$, $p=0.7$, and $D_m=50$.}
 \label{pi_opt_la=0.4_ga=0.3}
\end{figure}

Fig. \ref{pi_opt_la=0.4_ga=0.3} and Fig. \ref{pi_opt_la=0.4_ga=0.6} plot the optimal policy $\pi^*$ of a sub-MDP corresponding to the considered network with the increasing scheduling cost $C$.
Green circles indicate the states where it is optimal for the BS to remain idle, while red Asterisks indicate the states where scheduling is optimal.

For Fig. \ref{pi_opt_la=0.4_ga=0.3}, we set $\lambda=0.4$, $\gamma=0.3$, $p=0.7$, and $D_m=50$.
We can observe from Fig. \ref{pi_opt_la=0.4_ga=0.3} that $\mathcal{P}(C)$, i.e., the set of all green spots, monotonically expands as $C$ increases.
Additionally, $\mathcal{P}(C)=\emptyset$ when $C=0$ and $\mathcal{P}(C)=\mathcal{S}$ when $C$ is significantly large.
These phenomenons satisfy Definition \ref{DWA}, validating Whittle indexability of the sub-MDP of the considered scheduling problem. 

\begin{figure}[t]
	\centering
	\includegraphics[width=0.48\textwidth]{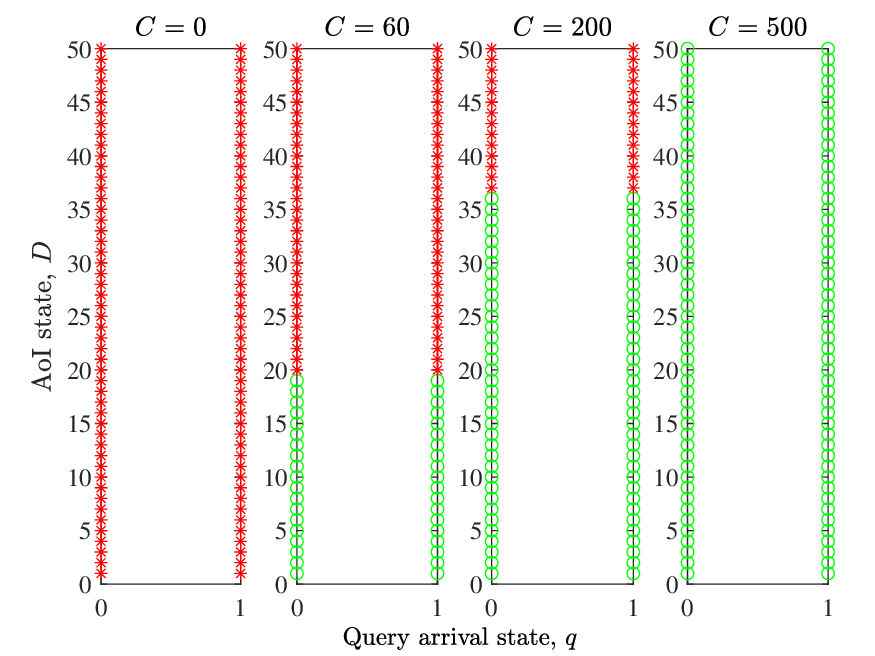}
	\caption{The optimal policy of a sub-MDP versus the scheduling cost $C$, where $\lambda=0.4$, $\gamma=0.6$, $p=0.7$, and $D_m=50$.}
 \label{pi_opt_la=0.4_ga=0.6}
\end{figure}

We set $\lambda=0.4$, $\gamma=0.6$, $p=0.7$, and $D_m=50$ for Fig. \ref{pi_opt_la=0.4_ga=0.6}, demonstrating the observations which are similar to those shown in Fig. \ref{pi_opt_la=0.4_ga=0.3}.
Furthermore, it is worth noticing that $H^*_0(C)$ is consistently equal to $H^*_1(C)$ as the increase of $C$.
This corresponds to Corollary \ref{EqualTh}, indicating that the optimal policy of the sub-MDP with $\lambda+\gamma=0.4+0.6=1$ always has identical $H^*_0(C)$ and $H^*_1(C)$.

\begin{figure}[t]
	\centering
	\includegraphics[width=0.48\textwidth]{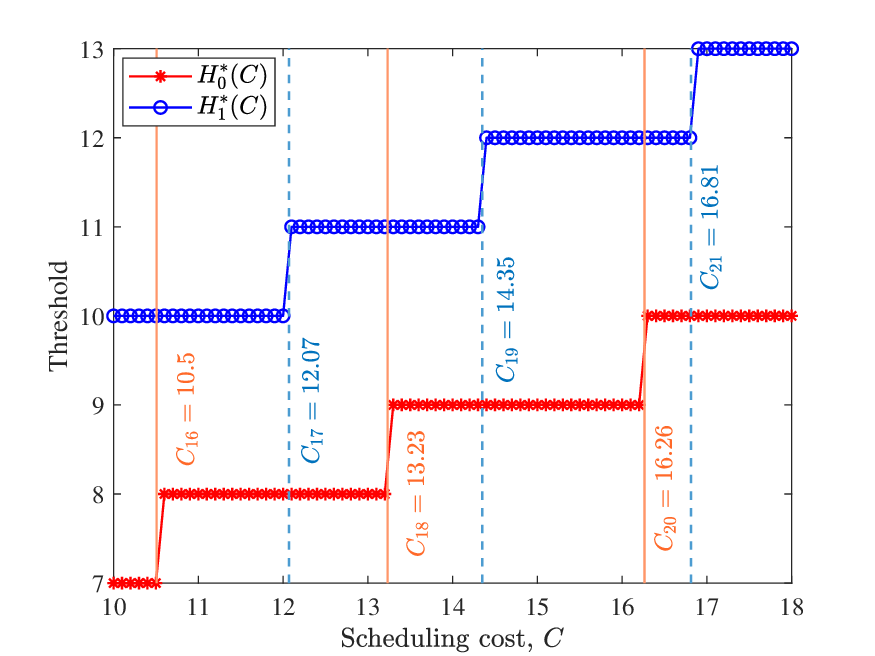}
	\caption{Thresholds versus the scheduling cost $C$, where $\lambda=0.5$, $\gamma=0.7$, $p=0.7$, and $D_m=50$.}
 \label{Threshold_v_C}
\end{figure}

In Fig. \ref{Threshold_v_C}, we depict the curves of the thresholds $H^*_0(C)$ and $H^*_1(C)$ of $\pi^*$ of a sub-MDP versus the scheduling cost $C$ with parameters given at the bottom of the figure.
We also calculate the values of distinct Whittle indices, denoted by $C_w$, of the sub-MDP within the range of $C$ shown in Fig. \ref{Threshold_v_C} and mark the locations of them in the figure.
In particular, according to the results of the proposed algorithm, $C_{16}$ is Whittle index of state $\langle 7,0 \rangle$, $C_{17}$ is Whittle index of state $\langle 10,1 \rangle$, $C_{18}$ is Whittle index of state $\langle 8,0 \rangle$, $C_{19}$ is the index of state $\langle 11,1 \rangle$, $C_{20}$ is the index of state $\langle 9,0 \rangle$, and $C_{21}$ is the index of state $\langle 12,1 \rangle$.
Recall that, built upon \eqref{TDWA}, Whittle index $C' = W(D,q)$ implies that, when the query state is $q$, $D$ and $D+1$ are equally optimal as thresholds for $\pi^*$ of the sub-MDP applying $C'$, which is consistent with the distribution of $C_w$ shown in Fig. \ref{Threshold_v_C}.
Therefore, Fig. \ref{Threshold_v_C} confirms the effectiveness of Algorithm~\ref{AlgErrorProne}.

\subsection{Comparisons with Baseline Policies}
We evaluate the performance of the proposed Whittle index policy by comparing it with baseline policies, shown below:

\begin{itemize}

    \item \textit{AoI-oriented Whittle index policy}: We introduce this policy from \cite{sun2022age}.
    This Whittle index policy is developed to optimize the long-term expected AoI of a network, where the network structure is the same as ours and the BS considers the latest received information from all nodes in all time slots.

    \item \textit{Greedy policy}: In each time slot, the BS adopting the greedy policy selects the scheduling action minimizing the expected immediate network-wide QAoI in the next time slot.
    Specifically, in each slot, the greedy policy ranks all nodes in descending order of their values of the term $D_i(t)p_i[\lambda_i(1-q_i(t))+(1-\gamma_i)q_i(t)]$, then schedules the first $M$ nodes from this sorted list.

\end{itemize}

\begin{figure}[t]
	\centering
	\includegraphics[width=0.48\textwidth]{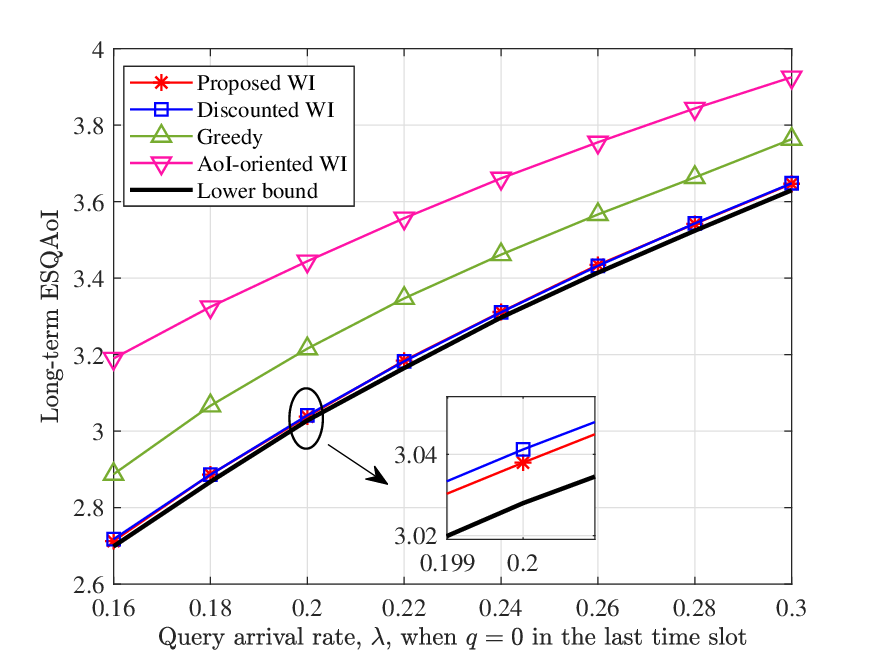}
	\caption{ESQAoI versus the query arrival rate if no query arriving in the last slot, where $\lambda_i=\lambda$, $p_i=0.8,\forall i$, $N=50$, $M=5$, and $D_m=50$.}
 \label{QAoI_v_la}
\end{figure}

We also introduce \textit{Whittle index policy obtained by \cite[Algorithm 1]{Akbarzadeh_Mahajan_2022}} and the lower bound of the ESQAoI performance achieved by solving problem \eqref{PR} to verify the effectiveness of the proposed policy. 
Recall that this Whittle index policy produced via \cite[Algorithm 1]{Akbarzadeh_Mahajan_2022} is adaptive to the sub-MDP under the expected discounted reward.
Furthermore, the computational complexity \cite[Algorithm 1]{Akbarzadeh_Mahajan_2022} is $\mathcal{O}(D_m^3)$, which is larger than Algorithm \ref{AlgErrorProne}. 
In the following, we call this type of strategy ``discounted Whittle index policy''. 
By Proposition \ref{PropVF}, we obtain the discounted Whittle index policy that can be approximately applied to our problem by setting $\beta$ close to $1$.

Fig. \ref{QAoI_v_la} shows the ESQAoI of the proposed, discounted, AoI-oriented Whittle index policies and the greedy policy are plotted against the query arrival rate when $q=0$ in the last time slot with parameters provided at the bottom of the figure.
The query arrival rate when $q=1$ in the last time slot of each node is uniformly set as one of $\{0.01,0.05,0.2,0.45,0.8\}$.
Further, we set $\beta=0.99$ to calculate the discounted Whittle indices.
As shown in Fig. \ref{QAoI_v_la}, all curves increase as $\lambda$ increases.
This is intuitive, as a large $\lambda$ increases the chance of each node being queried, causing the monitor to use the received packet in more time slots when the immediate AoI is high.
Additionally, Fig. \ref{QAoI_v_la} demonstrates that the performance of the discounted Whittle index policy is comparable to that of the proposed Whittle index policy.
This can be explained by Proposition \ref{PropVF}, implying that the sets of Whittle indices for $\mathcal{M}$ and $\mathcal{M}_{\beta}$ converge as $\beta$ is close to $1$.
The proposed policy is also shown to slightly outperform the discounted Whittle index policy, as the Whittle indices computed by the proposed algorithm are more accurate.
Across the range of $\lambda$, the consistently narrow gap between the proposed Whittle index policy and the lower bound highlights the asymptotical optimality of the proposed policy.
Moreover, we observe that the greedy policy performs worse compared to both the proposed and discounted Whittle index policies.
This occurs because the greedy policy solely guarantees optimality in the next time slot.
The AoI-oriented Whittle index policy exhibits the worst performance since this policy does not use the information of the query arrival states of the monitors.

\section{Conclusion}
In this paper, we studied the QAoI-oriented scheduling problem for a wireless multiuser uplink network with Markovian query arrivals at the base station (BS) side.
The decision-making problem was formulated as a Markov decision process (MDP), which is of prohibitively high complexity to be solved optimally by the standard method.
To address this issue, we applied Whittle index approach by initially relaxing the scheduling constraint of the original problem and decoupling the relaxed problem into multiple sub-MDPs.
We then proved that a sub-MDP is unichain, confirming that the convergence of the optimal policy under the expected discounted reward to that under the expected average reward.
Subsequently, we identified the threshold structure of the sub-MDP's optimal policy, which enabled the validation of Whittle indexability of the sub-MDP.
We then found the mathematical relationships between different Whittle indices of a sub-MDP by leveraging the convergence property of the sub-MDP. 
Drawn upon this, we put forth an efficient algorithm to compute the Whittle index of the formulated problem.
Subsequently, we leveraged the structural simplicity of the sub-MDP under error-free channels to devise an algorithm with further lower computational complexity for these scenarios.
Numerical results confirmed the accuracy of our theoretical analysis and demonstrated the superiority of the proposed Whittle index policy over the baseline policies.
Furthermore, the near-optimal performance of the proposed Whittle index policy was shown.

\bibliographystyle{IEEEtran}
\bibliography{QRef}

\appendices

\section{Proof of Proposition \ref{PropVF}}\label{Proof_PropVF}

Following \cite[Prop. 4.2.1, Prop. 4.2.6]{bertsekas2012dynamic}, we can prove Proposition \ref{PropVF} by showing that for any two states $s,s'\in\mathcal{S}$ of the sub-MDP, there exists a stationary deterministic policy $\pi_{s,s'}$ such that
\begin{equation}
    \Pr[s(\tau+1)=s'|s(1)=s,\pi_{s,s'}]>0
\end{equation}
for some $\tau\in\mathbb{Z}^+$.

Given an initial query arrival state $q(1)$, let $\mu_0(\tau+1)$ and $\mu_1(\tau+1)$ denote the probability of $q(\tau+1)=0$ and $q(\tau+1)=1$, i.e., the query arrival state equal to $0$ and that equal to $1$ after $\tau$ time slots.
Based on the Markov chain of $q$, its transition matrix is given by
\begin{equation}
    \mathbf{P}=
    \begin{bmatrix}
    1-\lambda & \lambda\\
    \gamma & 1-\gamma
    \end{bmatrix}.
\end{equation}
Clearly, we have $\Pr[q(\tau+1)=q'|q(1)=q]=\mathbf{P}^{\tau}_{q+1,q'+1}$.
By \cite{doi:10.1057/palgrave.jors.2600329}, $\mathbf{P}$ is diagonalizable and its two eigenvalues are $\alpha_1=1-\lambda-\gamma$ and $\alpha_2=1$.
Therefore, we can find an invertible matrix $\mathbf{Q}=[
\bm{\nu}_1 \: \bm{\nu}_2]$ such that $\mathbf{P}=\mathbf{Q}\mathbf{A}\mathbf{Q}^{-1}$, where $\bm{\nu}_1,\bm{\nu}_2$ are the eigenvectors corresponding to $\alpha_1$ and $\alpha_2$, respectively, and $\mathbf{A}=\text{diag}[\alpha_1,\alpha_2]$.

We now derive $\mathbf{Q}$ and $\mathbf{Q}^{-1}$.
Solving 
\begin{equation}
    (\mathbf{P}-\alpha_1\mathbf{I})\bm{\nu}_1=
    \begin{bmatrix}
        \gamma & \lambda\\
        \gamma & \lambda
    \end{bmatrix}\bm{\nu}_1=\mathbf{0},
\end{equation}
we can obtain that $\bm{\nu}_1$ can be $[-\frac{\lambda}{\gamma},1]^T$, and solving 
\begin{equation}
    (\mathbf{P}-\alpha_2\mathbf{I})\bm{\nu}_2=
    \begin{bmatrix}
        -\lambda & \gamma\\
        \lambda & -\gamma
    \end{bmatrix}\bm{\nu}_2=\mathbf{0},
\end{equation}
yields that $\bm{\nu}_2$ can be $[1,1]^T$.
Thus, we have 
\begin{equation}
    \mathbf{Q}=
    \begin{bmatrix}
         -\frac{\lambda}{\gamma}& 1\\
         1 & 1
    \end{bmatrix},
    \mathbf{Q}^{-1}=
    \begin{bmatrix}
        -\frac{\gamma}{\lambda+\gamma} & \frac{\gamma}{\lambda+\gamma}\\
        \frac{\gamma}{\lambda+\gamma} & \frac{\lambda}{\lambda+\gamma}
    \end{bmatrix}.
\end{equation}
After some manipulations, 
\begin{equation}
    \mathbf{P}^{\tau}=\mathbf{Q}\mathbf{A}^{\tau}\mathbf{Q}^{-1}=
    \begin{bmatrix}
\frac{\lambda(1-\lambda-\gamma)^{\tau}+\gamma}{\lambda+\gamma} & \frac{-\lambda(1-\lambda-\gamma)^{\tau}+\lambda}{\lambda+\gamma}\\
\frac{-\gamma(1-\lambda-\gamma)^{\tau}+\gamma}{\lambda+\gamma} & \frac{\gamma(1-\lambda-\gamma)^{\tau}+\lambda}{\lambda+\gamma}
\end{bmatrix}.
\end{equation}
Notice that all elements of $\mathbf{P}^{\tau}$ are non-zero for $\lambda,\gamma\in(0,1)$, indicating that $\Pr[q(\tau+1)=q'|q(1)=q]>0$ for any $q$, $q'$, and $\tau$.

Recall the evolution of AoI and query are independent of each other.
Therefore, when $s(1)=\langle D,q\rangle$ and $s(\tau+1)=\langle D',q'\rangle$ satisfy $D<D'\le D_m$, we can let $\pi_{s,s'}$ take passive action for all states, leading to
\begin{equation}
    \begin{split}
        &\Pr[s(\tau+1)=s'|s(1)=s,\pi_{s,s'}]\\
        &=\Pr[q(\tau+1)=q'|q(1)=q]>0
    \end{split}
\end{equation}
for $\tau=D'-D$.
On the other hand, when $D$ and $D'$ satisfy $D\ge D'$, we can let $\pi_{s,s'}$ take active action when the AoI state is not smaller than $D$, and take passive action otherwise, resulting in
\begin{equation}
    \begin{split}
        &\Pr[s(\tau+1)=s'|s(1)=s,\pi_{s,s'}]\\
        &=p(1-p)^{k-1}\Pr[q(\tau+1)=q'|q(1)=q]>0
    \end{split}
\end{equation}
for $\tau=k+D'-1$, $k\in\mathbb{Z}^+$.
This completes the proof.

\section{Proof of Lemma \ref{Vp}}\label{Proof_Vp}

We can prove Lemma \ref{Vp} by the induction.
    
    Recall $V_{\beta,1}(D,q)=0$ for all $s\in\mathcal{S}$, implying that $V_{\beta,1}(D,q)$ is non-decreasing with regard to $D$.
    Suppose $V_{n,\beta}(D,q)$ is non-decreasing with respect to $D$ given $q$, i.e.,
    \begin{equation}
        \frac{\partial^- V_{n,\beta}(D,q)}{\partial D}
        \ge 0, \forall q\in\{0,1\}.
    \end{equation}
    As such,
    \begin{equation}\label{PDD+1}
        \frac{\partial^- V_{n,\beta}(D+1,q)}{\partial D}=\frac{\partial^- V_{n,\beta}(D+1,q)}{\partial (D+1)}\ge 0,
    \end{equation}
    for all $q\in\{0,1\},D\in\{1,2,\cdots,D_m-1\}$.
    For $n+1$, by \eqref{Q00}, \eqref{Q01}, and \eqref{PDD+1}, we can easily show that $Q_{n+1,\beta}(D,0;0)$ and $Q_{n+1,\beta}(D,0;1)$ are non-decreasing with regard to $D$, thus $V_{n+1,\beta}(D,0)=\min\{Q_{n+1,\beta}(D,0;0),Q_{n+1}(D,0;1)\}$ is non-decreasing with respect to $D$.
    Similarly, $V_{n+1,\beta}(D,1)$ is also non-decreasing with respect to $D$.
    In light of this, we can conclude this non-decreasing property holds for $n+1$.
    
    Hence, $V_{n,\beta}(D,q)$ is an non-decreasing function of $D$ for all $n$, yielding that $V^*_{\beta}(D,q)=\lim_{n\to\infty}V_{n,\beta}(D,q)$ is non-decreasing with respect to $D$.

\section{Proof of Theorem \ref{ThP}}\label{Proof_ThP}

Consider $V^*_{\beta}(D,0)$ with arbitrary $D\in\{1,2,\cdots,D_m\}$.
    Assume that $\pi^*_{\beta}$ chooses to schedule for the state $s=\langle D,0\rangle$, i.e.,
    \begin{equation}\label{CompareD0}
        \begin{split}
            &Q^*_{\beta}(D,0;0)-Q^*_{\beta}(D,0;1)\\
            &=\beta\left[p(1-\lambda)V^*_{\beta}(\bar{D},0)+p\lambda V^*_{\beta}(\bar{D},1)\right.\\
            &\left.-p(1-\lambda)V^*_{\beta}(1,0)-p\lambda V^*_{\beta}(1,1)\right]-C\ge 0.
        \end{split}
    \end{equation}
    Then, based on Lemma \ref{Vp}, for any $D'\ge D$, 
    \begin{equation}
        \begin{split}
            &Q^*_{\beta}(D',0;0)-Q^*_{\beta}(D',0;1)\\
            &=\beta\left[p(1-\lambda)V^*_{\beta}(D'+1,0)+p\lambda V^*_{\beta}(D'+1,1)\right.\\
            &\left.-p(1-\lambda)V^*_{\beta}(1,0)-p\lambda V^*_{\beta}(1,1)\right]-C\\
            &\ge\beta\left[p(1-\lambda)V^*_{\beta}(\bar{D},0)+p\lambda V^*_{\beta}(\bar{D},1)\right.\\
            &\left.-p(1-\lambda)V^*_{\beta}(1,0)-p\lambda V^*_{\beta}(1,1)\right]-C\\
            &= Q^*_{\beta}(D,0;0)-Q^*_{\beta}(D,0;1)\ge 0,
        \end{split}
    \end{equation}
    suggesting that the optimal policy also selects to schedule for any states with ages larger than $D$ given $q=0$.
    Similarly, we can obtain that the optimal policy takes active action for any states with ages larger than $D$ given $q=1$ if $\pi^*_{\beta}(D,1)=1$, i.e, $Q^*_{\beta}(D,1;0)-Q^*_{\beta}(D,1;1)\ge 0$.

    Clearly, the above properties hold if and only if the optimal policy $\pi^*_{\beta}$ has the threshold structure.

\section{Proof of Corollary \ref{EqualTh}}\label{Proof_EqualTh}

For state $s=\langle D,0 \rangle$ with arbitrary $D\in\{1,2,\cdots,D_m\}$, recall \eqref{CompareD0}.
    Further, for $s=\langle D,1 \rangle$, we have
    \begin{equation}
    \begin{split}
        &Q^*_{\beta}(D,1;0)-Q^*_{\beta}(D,1;1)\\
        &=\beta p\left[\gamma V^*_{\beta}(\bar{D},0)+(1-\gamma) V^*_{\beta}(\bar{D},1)\right.\\
            &\left.-\gamma V^*_{\beta}(1,0)-(1-\gamma) V^*_{\beta}(1,1)\right]-C.
    \end{split}
    \end{equation}
    Obviously, when $\lambda+\gamma=1$, 
    \begin{equation}\label{compare}
        Q^*_{\beta}(D,0;0)-Q^*_{\beta}(D,0;1)=Q^*_{\beta}(D,1;0)-Q^*_{\beta}(D,1;1),
    \end{equation}
    for all $D$.

    Assume that $H^*_0(C)\neq H^*_1(C)$.
    Then, following the threshold structure, for 
    \begin{equation*}
        \min\{H^*_0(C),H^*_1(C)\}\le D<\max\{H^*_0(C),H^*_1(C)\},
    \end{equation*}
    we have $\pi^*_{\beta}(D,0)\neq \pi^*_{\beta}(D,1)$, i.e.,
    \begin{equation}
        Q^*_{\beta}(D,0;0)-Q^*_{\beta}(D,0;1)\neq Q^*_{\beta}(D,1;0)-Q^*_{\beta}(D,1;1),
    \end{equation}
    which is contradicts \eqref{compare}.
    Therefore, we must have $H^*_0(C)=H^*_1(C)$ for $\pi^*_{\beta}$ by Proposition \ref{PropVF}.

\section{Proof of Proposition \ref{VAT}}\label{Proof_VAT}

Define two MDPs, $\mathcal{M}_0$ and $\mathcal{M}_1$, with components identical to those of the sub-MDP under consideration.
    $\mathcal{M}_0$ and $\mathcal{M}_1$ have exactly the same initial state $s_0(1)=s_1(1)=\langle D,q\rangle$, activation cost $C$, and the state transition probability, $\lambda$, $\gamma$, and $p$.
    Furthermore, the two MDPs are under the same policy $\pi\in\Pi_h$ in $t\ge 2$.
    However, $\mathcal{M}_0$ and $\mathcal{M}_1$ have one difference, $a_0(1)=0$ while $a_1(1)=1$.
    Clearly, the expected discounted total active times of $\mathcal{M}_0$ and $\mathcal{M}_1$ are equivalent to $A^{\pi}_{\beta}(s;0)$ and $A^{\pi}_{\beta}(s;1)$, respectively.

    We then couple $\mathcal{M}_0$ and $\mathcal{M}_1$ together.
    Specifically, assume that for $t\ge 1$, the $j$th scheduling of $\mathcal{M}_0$ will have the same transmission event (success or failure) as that of the $j$th scheduling of $\mathcal{M}_1$.
    Moreover, the query arrivals of $\mathcal{M}_0$ and $\mathcal{M}_1$ are coupled as: the query arrival event of $\mathcal{M}_0$ in slot $t$ is the same as the query arrival event of $\mathcal{M}_1$ for all $t$.
    These two couplings do not change state transition probability since the evolution manners of $D$ and $q$ of both MDPs are unchanged.

    By \cite[Lemma 9, Lemma 10, Theorem 11]{zhoueasier}, the coupled MDPs $\mathcal{M}_0$, $\mathcal{M}_1$ with two-dimensional states, which have the same evolution on one state component, i.e., $q$-component, and the threshold structure in the other state component, i.e., $D$-component, satisfy that the expected discounted total time of $\mathcal{M}_1$ is not less than that of $\mathcal{M}_0$.
    Equivalently, the sub-MDP holds that $A^{\pi}_{\beta}(s;1)\ge A^{\pi}_{\beta}(s;0),\forall s\in\mathcal{S}$ and $\forall \pi\in\Pi_h$.

\section{Proof of Theorem \ref{CalWI}}\label{Proof_CalWI}

    Define $N_{\beta,w}$, $C_{\beta,w}$, $\mathcal{R}_{\beta,w}$, and $\Gamma_{\beta,w}$, which have the same definitions as $N_w$, $C_w$, $\mathcal{R}_w$, and $\Gamma_w$, but applied for $\mathcal{M}_{\beta}$.

    Besides, for $w\in\{0,1,\cdots,N_{\beta,w}-1\}$, state $y\in\mathcal{S}\backslash\mathcal{R}_{\beta,w}$, we define $\eta_{\beta,w}=\pi^{\mathcal{R}_{\beta,w}}$, $\eta_{\beta,(w,y)}=\pi^{\mathcal{R}_{\beta,w}\cup\{y\}}$ for brevity, and
    \begin{equation}
        \mathcal{C}_{\beta}(w,y)=\{s\in\mathcal{S}|A^{\eta_{\beta,w}}_{\beta}(s)\neq A^{\eta_{\beta,(w,y)}}_{\beta}(s)\},
    \end{equation}

    For $w\in\{0,1,\cdots,N_{\beta,w}-1\}$, given $C_{\beta,0}$ to $C_{\beta,w}$, \cite[Theorem 2]{Akbarzadeh_Mahajan_2022} shows that $\mathcal{M}_{\beta}$ has a property:

    There exists $\mathcal{Y}\subseteq\mathcal{S}\backslash\mathcal{R}_{\beta,w}$ such that $\mathcal{C}(w,y)$ is non-empty for $y\in\mathcal{Y}$, and exists a $C'_{\beta}>C_{\beta,w}$ such that

    \begin{equation}\label{Ceq}
    \begin{split}
        V^{\eta_{\beta,w}}_{\beta}(s)&\!=\!J^{\eta_{\beta,w}}_{\beta}(s)+C'_{\beta}A^{\eta_{\beta,w}}_{\beta}(s)\\
        &\!=\!J^{\eta_{\beta,(w,y)}}_{\beta}(s)\!+\!C'_{\beta}A^{\eta_{\beta,(w,y)}}_{\beta}(s)\!=\!V^{\eta_{\beta,(w,y)}}_{\beta}(s)\\
        &\le V^{\pi\neq \eta_{\beta,w},\eta_{\beta,(w,y)}}_{\beta}(s),
    \end{split}       
    \end{equation}
    and
    \begin{equation}\label{CO}
        C'_{\beta}=\min_{y\in\mathcal{Y}}\min_{s\in\mathcal{C}(w,y)}\frac{J^{\eta_{\beta,(w,y)}}_{\beta}(s)-J^{\eta_{\beta,w}}_{\beta}(s)}{A^{\eta_{\beta,w}}_{\beta}(s)-A^{\eta_{\beta,(w,y)}}_{\beta}(s)}.
    \end{equation}

    Notice that the property shown in \eqref{Ceq} and \eqref{CO} holds for all $\beta$, thereby
    also holds as $\beta\to 1$.
    Then, define $\eta_w=\pi^{\mathcal{R}_w}$, $\eta_{w,y}=\pi^{\mathcal{R}_w\cup\{y\}}$ for brevity, and
    \begin{equation}
        \mathcal{C}(w,y)=\{s\in\mathcal{S}|A^{\eta_w}\neq A^{\eta_{w,y}}\},
    \end{equation}
    for $\mathcal{M}$.
    For $w\in\{0,1,\cdots,N_W-1\}$, given $C_0$ to $C_w$, we have $\mathcal{C}(w,y)=\mathcal{S}$ since $A^{\pi}$ is unchanged with any initial state $s$.
    Consequently, there exists a $C'> C_w$ such that
    \begin{equation}\label{Ceq1}
    \begin{split}
        R^{\eta_w}&=\lim_{\beta\to 1}(1-\beta)[J^{\eta_{\beta,w}}_{\beta}(s)+C'A^{\eta_{\beta,w}}_{\beta}(s)]\\
        &=\lim_{\beta\to 1}(1-\beta)[J^{\eta_{\beta,(w,y)}}_{\beta}(s)+C'A^{\eta_{\beta,(w,y)}}_{\beta}(s)]\\
        &=R^{\eta_{w,y}}
        \le R^{\pi\neq \eta_w,\eta_{w,y}},\forall s\in\mathcal{S},
    \end{split}       
    \end{equation}
    and
    \begin{equation}\label{CO1}
    \begin{split}
        C'&=\min_{y\in\mathcal{S}\backslash\mathcal{R}_w}\lim_{\beta\to 1}(1-\beta)\frac{J^{\eta_{\beta,(w,y)}}_{\beta}(s)-J^{\eta_{\beta,w}}_{\beta}(s)}{A^{\eta_{\beta,w}}_{\beta}(s)-A^{\eta_{\beta,(w,y)}}_{\beta}(s)}\\
        &=\min_{y\in\mathcal{S}\backslash\mathcal{R}_w}\frac{J^{\eta_{w,y}}-J^{\eta_w}}{A^{\eta_w}-A^{\eta_{w,y}}}=\min_{y\in\mathcal{S}\backslash\mathcal{R}_w}\zeta[\eta_{w,y},\eta_w].
    \end{split}
    \end{equation}

    Eq. \eqref{Ceq1} reveals that both $\eta_w$ and $\eta_{w,y}$ are the optimal policies, indicating that $C'$ satisfies Definition \ref{DWA}, and thus is Whittle index of state $y$.
    Thus, all $y$ that can attain the minimization of \eqref{CO1} constitute $\Gamma_{w+1}$.
    By Definition \ref{DWA}, we can easily conclude that $\eta_w$ is an optimal policy if and only if $C_w\le C \le C_{w+1}$, which means $C_{w+1}=C'$ and $\eta_w=\pi^{\mathcal{R}_w}\in\Pi_h$. 
    In light of this, we can obtain \eqref{Cnexteq}.

    This completes the proof.

\section{Proof of Proposition \ref{TransMatrix}}\label{Proof_TransMatrix}

By the state transitions in Fig. \ref{MCTH}, we can readily obtain \eqref{P1}, \eqref{P2}, and \eqref{P3}.

    \begin{figure}[t]
    	\centering
    	\includegraphics[width=0.48\textwidth]{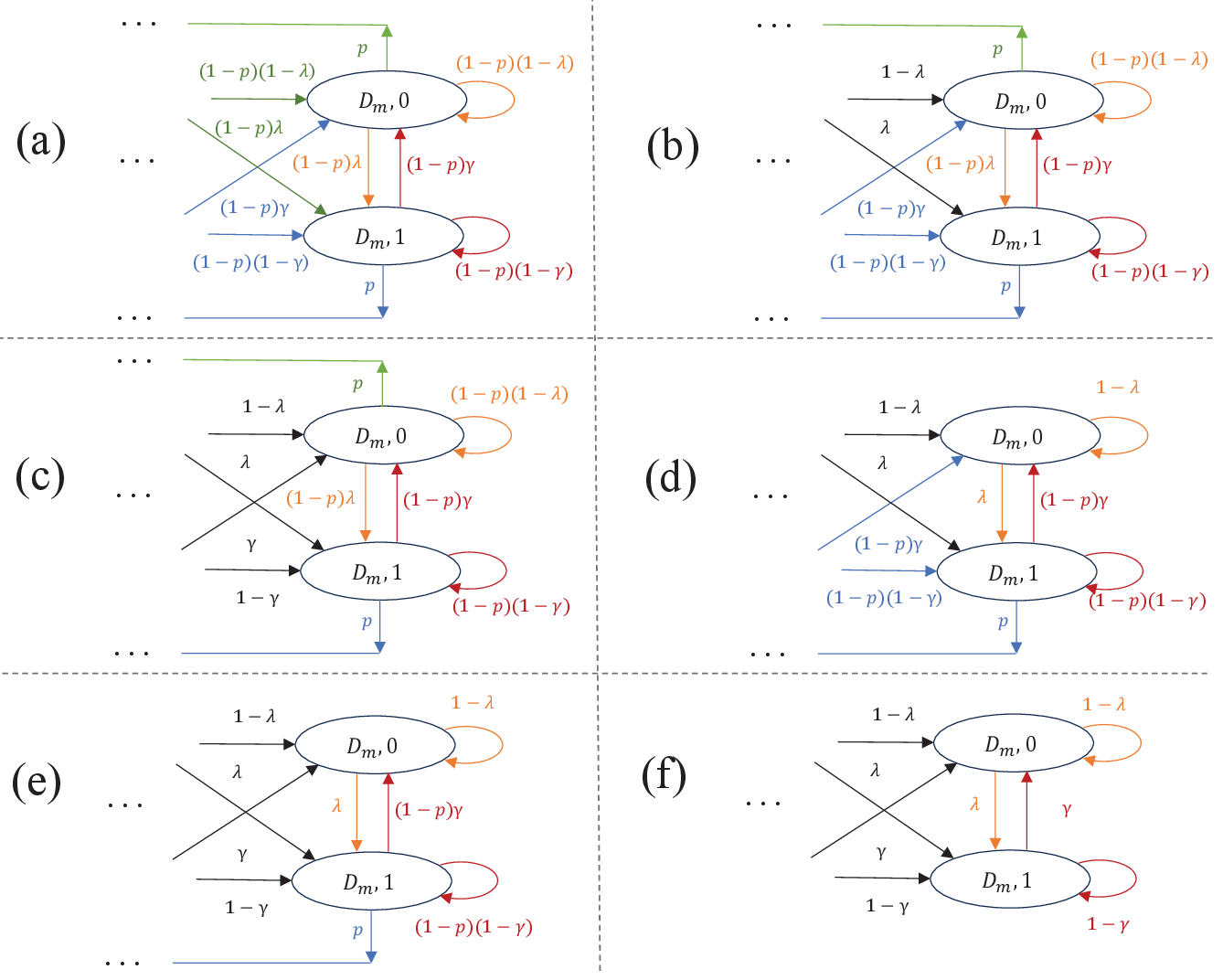}
    	\caption{An example of 6 different cases of state transition from $\bm{\mu}_{D_m-1}$ to  $\bm{\mu}_{D_m}$ of a sub-MDP applying a policy $\pi(H_0;H_1)$ with $H_0\ge H_1$.}
     \label{MCTHDmP}
    \end{figure}
    
    To investigate $\mathbf{P}_4$, we need to consider 6 different cases of state transition from $\bm{\mu}_{D_m-1}$ to $\bm{\mu}_{D_m}$ of the sub-MDP with $H_0>H_1$, which is illustrated in Fig. \ref{MCTHDmP}.

    The state transitions in Subplot (a) of Fig. \ref{MCTHDmP} shows the case when $H_1\le H_0<D_m$.
    We can observe that
    \begin{equation}\label{Dma}
        \bm{\mu}_{D_m}=\mathbf{P}_3(\bm{\mu}_{D_m-1}+\bm{\mu}_{D_m}).
    \end{equation}
    By the state transitions in Subplot (b) of Fig. \ref{MCTHDmP} depicting the case that $H_1< H_0=D_m$, we can conclude
    \begin{equation}\label{Dmb}
        \bm{\mu}_{D_m}=\mathbf{P}_2\bm{\mu}_{D_m-1}+\mathbf{P}_3\bm{\mu}_{D_m}.
    \end{equation}
    By the state transitions in Subplot (c) of Fig. \ref{MCTHDmP} showing the case that $H_0=H_1=D_m$, we have
    \begin{equation}\label{Dmc}
        \bm{\mu}_{D_m}=\mathbf{P}_1\bm{\mu}_{D_m-1}+\mathbf{P}_3\bm{\mu}_{D_m}.
    \end{equation}
    By the state transitions in Subplot (d) of Fig. \ref{MCTHDmP} showing the case that $H_1<D_m,H_0=D_m+1$, we can write
    \begin{equation}\label{Dmd}
        \bm{\mu}_{D_m}=\mathbf{P}_2(\bm{\mu}_{D_m-1}+\bm{\mu}_{D_m}).
    \end{equation}
    By the state transitions in Subplot (e) of Fig. \ref{MCTHDmP} plotting the case that $H_1=D_m,H_0=D_m+1$, we state
    \begin{equation}\label{Dme}
        \bm{\mu}_{D_m}=\mathbf{P}_1\bm{\mu}_{D_m-1}+\mathbf{P}_2\bm{\mu}_{D_m}.
    \end{equation}
    By the state transitions in Subplot (f) of Fig. \ref{MCTHDmP} depicting the case that $H_0=H_1=D_m+1$, we have
    \begin{equation}\label{Dmf}
        \bm{\mu}_{D_m}=\mathbf{P}_1(\bm{\mu}_{D_m-1}+\bm{\mu}_{D_m}).
    \end{equation}

    Eqs. \eqref{Dma}-\eqref{Dmf} remain valid when $H_0 \le H_1$. After performing some manipulations on \eqref{Dma} to \eqref{Dmf}, we can derive \eqref{P4}.
    Note that $\mathbf{P}_1$ is a stochastic matrix and $\mathbf{P}_3$ is a sub-stochastic matrix, thus they have spectral radiuses less than $1$.
    Moreover, $\det(\mathbf{I}-\mathbf{P}_2)=\lambda p$ when $H_0>H_1$ and $\det(\mathbf{I}-\mathbf{P}_2)=\gamma p$ when $H_0<H_1$.
    Accordingly, $\mathbf{I}-\mathbf{P}_d$ is invertible for $d=1,2,3$ in the cases considered in \eqref{P4}, and thus \eqref{P4} is feasible.

\section{Useful Expressions for Implementing Algorithm \ref{AlgErrorProne} When Channel is Error-prone}\label{Expression_Error_Channel}
Considering a sub-MDP applying policy $\pi(H_0;H_1)$ with $\lambda$, $\gamma$, $D_m$, and $p<1$.
By the structure of the sub-MDP shown in Fig. \ref{MCTH}, we have
\begin{equation}
    \mu_{D,0}\!=\!\frac{(1\!-\!\lambda\!-\!\gamma)^{H\!-\!1}(\lambda\mu_{1,0}\!-\!\gamma\mu_{1,1})\!+\!\gamma(\mu_{1,0}\!+\!\mu_{1,1})}{\lambda\!+\!\gamma},
\end{equation}
\begin{equation}
    \mu_{D,1}\!=\!\frac{(1\!-\!\lambda\!-\!\gamma)^{H\!-\!1}\!(-\lambda\mu_{1,0}\!+\!\gamma\mu_{1,1})\!+\!\lambda(\mu_{1,0}\!+\!\mu_{1,1})}{\lambda+\gamma},
\end{equation}
for $D\le H$.
As such, $\mu_{D,0}+\mu_{D,1}=\mu_{1,0}+\mu_{1,1}$ for $D\le H$.

In addition, we use inequalities involving $H_0$ and $H_1$ to represent different cases for expressions of the linear equations, $J^{\pi(H_0;H_1)}$, and $A^{\pi(H_0;H_1)}$. 
Consider two inequalities of two cases, they have the same structure but the positions of $H_0$ and $H_1$ are swapped.
We call these two cases symmetrical to each other. 

By $\mu_{D,0}+\mu_{D,1}=\mu_{1,0}+\mu_{1,1}$ for $D\le H$, after some mathematical manipulations, we can express the corresponding linear equations, $J^{\pi(H_0;H_1)}$, and $A^{\pi(H_0;H_1)}$ in different cases as follows.

Consider $H_1\le H_0$.
For $H_1\le H_0<D_m$, the linear equations are
\begin{equation}
\begin{split}
1&=H(\mu_{1,0}+\mu_{1,1})\\
&+\left\|\left(\sum_{i=1}^{\hat{H}-H}\mathbf{P}_2^i\mathbf{P}_{1}^{H-1}+\sum_{i=1}^{D_m-\hat{H}-1}\mathbf{P}_3^i\mathbf{P}_2^{\hat{H}-H}\mathbf{P}_1^{H-1}\right.\right.\\
&\left.\left.+\mathbf{P}_4\mathbf{P}_3^{D_m-\hat{H}-1}\mathbf{P}_2^{\hat{H}-H}\mathbf{P}_1^{H-1}\right)\bm{\mu}_1 \right\|_1,
\end{split}
\end{equation}
\begin{equation}
\begin{split}
    \mu_{1,0}&\!=\!(1-\lambda)p\left(\sum_{i=0}^{D_m-\hat{H}-1}\mathbf{P}_3^i\mathbf{P}_2^{\hat{H}-H}\mathbf{P}_1^{H-1}\right.\\
    &\!\left.+\!\mathbf{P}_4\mathbf{P}_3^{D_m-\hat{H}-1}\mathbf{P}_2^{\hat{H}-H}\mathbf{P}_1^{H-1}\right)_{1\cdot}\bm{\mu}_1\\
    &\!+\!\gamma p\!\left(\sum_{i=0}^{\hat{H}-H}\mathbf{P}_2^i\mathbf{P}_{1}^{H-1}\!+\!\sum_{i=1}^{D_m-\hat{H}-1}\mathbf{P}_3^i\mathbf{P}_2^{\hat{H}-H}\mathbf{P}_1^{H-1}\right.\\
    &\left.+\mathbf{P}_4\mathbf{P}_3^{D_m-\hat{H}-1}\mathbf{P}_2^{\hat{H}-H}\mathbf{P}_1^{H-1}\right)_{2\cdot}\bm{\mu}_1,
\end{split}
\end{equation}
and $J^{\pi(H_0;H_1)}$, $A^{\pi(H_0;H_1)}$ are given by
\begin{equation}
    \begin{split}
        J^{\pi(H_0;H_1)}&= \left[\sum_{i=0}^{H-1}(i+1)\mathbf{P}_{1}^i + \sum_{i=1}^{\hat{H}-H}(H+i)\mathbf{P}_2^i\mathbf{P}_{1}^{H-1}\right.\\
        &+\sum_{i=1}^{D_m-\hat{H}-1}(\hat{H}+i)\mathbf{P}_3^i\mathbf{P}_2^{\hat{H}-H}\mathbf{P}_1^{H-1}\\
        &\left.+D_m\mathbf{P}_4\mathbf{P}_3^{D_m-\hat{H}-1}\mathbf{P}_2^{\hat{H}-H}\mathbf{P}_1^{H-1}\right]_{2\cdot}\bm{\mu}_1,
    \end{split}
\end{equation}
\begin{equation}
    \begin{split}
        A^{\pi(H_0;H_1)}&=\left\|\left(\sum_{i=0}^{D_m-\hat{H}-1}\mathbf{P}_3^i\mathbf{P}_2^{\hat{H}-H}\mathbf{P}_1^{H-1}\right.\right.\\
        &\left.\left.+\mathbf{P}_4\mathbf{P}_3^{D_m-\hat{H}-1}\mathbf{P}_2^{\hat{H}-H}\mathbf{P}_1^{H-1}\right)\bm{\mu}_1\right\|_1\\
    &+\left(\sum_{i=0}^{\hat{H}-H}\mathbf{P}_2^i\mathbf{P}_{1}^{H-1}\right)_{2\cdot}\bm{\mu}_1.
    \end{split}
\end{equation}
For $H_1< H_0=D_m$, we have
\begin{equation}\label{lpSv2}
\begin{split}
1&=H(\mu_{1,0}+\mu_{1,1})\\
&\!+\!\left\|\left(\sum_{i=1}^{D_m\!-\!H\!-\!1}\mathbf{P}_2^i\mathbf{P}_{1}^{H-1}\!+\!\mathbf{P}_4\mathbf{P}_2^{D_m-H-1}\mathbf{P}_1^{H-1}\right)\bm{\mu}_1\right\|_1,
\end{split}
\end{equation}
\begin{equation}
\begin{split}
    &\mu_{1,0}=(1-\lambda)p\left(\mathbf{P}_4\mathbf{P}_2^{D_m-H-1}\mathbf{P}_1^{H-1}\right)_{1\cdot}\bm{\mu}_1\\
    &\!+\!\gamma p\!\left(\sum_{i=0}^{D_m\!-\!H\!-\!1}\mathbf{P}_2^i\mathbf{P}_{1}^{H-1}\!+\!\mathbf{P}_4\mathbf{P}_2^{D_m\!-\!H\!-\!1}\mathbf{P}_1^{H-1}\right)_{2\cdot}\bm{\mu}_1,
\end{split}
\end{equation}
and $J^{\pi(H_0;H_1)}$, $A^{\pi(H_0;H_1)}$ are given by
\begin{equation}\label{Jv2}
    \begin{split}
        J^{\pi(H_0;H_1)}&\!= \!\left[\sum_{i=0}^{H-1}(i+1)\mathbf{P}_{1}^i 
        \!+\!\sum_{i=1}^{D_m\!-\!H\!-\!1}(H\!+\!i)\mathbf{P}_2^i\mathbf{P}_{1}^{H-1}\right.\\
        &\left.+D_m\mathbf{P}_4\mathbf{P}_2^{D_m-H-1}\mathbf{P}_1^{H-1}\right]_{2\cdot}\bm{\mu}_1,
    \end{split}
\end{equation}
\begin{equation}
    \begin{split}
        A^{\pi(H_0;H_1)}&=\left\|\left(\mathbf{P}_4\mathbf{P}_2^{D_m-H-1}\mathbf{P}_1^{H-1}\right)\bm{\mu}_1\right\|_1\\
    &+\left(\sum_{i=0}^{D_m-H-1}\mathbf{P}_2^i\mathbf{P}_{1}^{H-1}\right)_{2\cdot}\bm{\mu}_1.
    \end{split}
\end{equation}
For $H_0=H_1=D_m$, we have
\begin{equation}\label{lpSv3}
\begin{split}
1=&(H-1)(\mu_{1,0}+\mu_{1,1})
+\left\|\left(\mathbf{P}_4\mathbf{P}_1^{D_m-2}\right)\bm{\mu}_1\right\|_1
\end{split}
\end{equation}
and
\begin{equation}
\begin{split}
    &\mu_{1,0}\!=\!(1\!-\!\lambda)p\left(\mathbf{P}_4\mathbf{P}_1^{D_m\!-\!2}\right)_{1\cdot}\!\bm{\mu}_1
    \!+\!\gamma p\left(\mathbf{P}_4\mathbf{P}_1^{D_m\!-\!2}\right)_{2\cdot}\!\bm{\mu}_1,
\end{split}
\end{equation}
and $J^{\pi(H_0;H_1)}$, $A^{\pi(H_0;H_1)}$ are given by
\begin{equation}\label{Jv3}
    \begin{split}
        J^{\pi(H_0;H_1)}&
       \! =\! \left[\sum_{i=0}^{D_m-2}(i+1)\mathbf{P}_{1}^i\! + \!D_m\mathbf{P}_4\mathbf{P}_1^{D_m-2}\right]_{2\cdot}\bm{\mu}_1,
    \end{split}
\end{equation}
\begin{equation}
    \begin{split}
        A^{\pi(H_0;H_1)}&=\left\|\left(\mathbf{P}_4\mathbf{P}_1^{D_m-2}\right)\bm{\mu}_1\right\|_1.
    \end{split}
\end{equation}
For $H_1 < D_m$ and $H_0 = D_m + 1$, the first linear equation is identical to \eqref{lpSv2}, while the second is given by
\begin{equation}
\begin{split}
    &\mu_{1,0}\!=\!\gamma p\left(\sum_{i=0}^{D_m\!-\!H\!-\!1}\mathbf{P}_2^i\mathbf{P}_{1}^{H-1}\!+\!\mathbf{P}_4\mathbf{P}_2^{D_m\!-\!H\!-\!1}\mathbf{P}_1^{H\!-\!1}\right)_{2\cdot}\!\bm{\mu}_1.
\end{split}
\end{equation}
$J^{\pi(H_0;H_1)}$ is identical to \eqref{Jv2}, while $A^{\pi(H_0;H_1)}=\mu_{1,0}/(\gamma p)$.
For $H_1=D_m$ and $H_0=D_m+1$, the first linear equation is identical to \eqref{lpSv3}, while the second one is presented as
\begin{equation}
\begin{split}
    &\mu_{1,0}=\gamma p\left(\mathbf{P}_4\mathbf{P}_1^{D_m-2}\right)_{2\cdot}\bm{\mu}_1,
\end{split}
\end{equation}
and $J^{\pi(H_0;H_1)}$ is identical to \eqref{Jv3}, while $A^{\pi(H_0;H_1)}=\mu_{1,0}/(\gamma p)$.
For $H_0=H_1=D_m+1$, we can easily obtain
\begin{equation}
    \begin{split}
        J^{\pi(H_0;H_1)}&=\frac{D_m\lambda}{\lambda+\gamma},
    \end{split}
\end{equation}
and $A^{\pi(H_0;H_1)}=0$.

On the other hand, we can easily conclude that all cases holding $H_0<H_1$ have the expressions of the first linear equation and $J^{\pi(H_0;H_1)}$ identical to that in the corresponding symmetrical cases.

Besides, in particular, for $H_0\le H_1<D_m$, we have
\begin{equation}
\begin{split}
    \mu_{1,0}&\!=\!(1\!-\!\lambda)p\!\left(\sum_{i=0}^{\hat{H}\!-\!H}\mathbf{P}_2^i\mathbf{P}_{1}^{H-1}\!+\!\sum_{i=1}^{D_m\!-\!\hat{H}\!-\!1}\mathbf{P}_3^i\mathbf{P}_2^{\hat{H}\!-\!H}\mathbf{P}_1^{H\!-\!1}\right.\\
    &\left.+\mathbf{P}_4\mathbf{P}_3^{D_m-\hat{H}-1}\mathbf{P}_2^{\hat{H}-H}\mathbf{P}_1^{H-1}\right)_{1\cdot}\bm{\mu}_1\\
    &+\gamma p\left(\sum_{i=0}^{D_m-\hat{H}-1}\mathbf{P}_3^i\mathbf{P}_2^{\hat{H}-H}\mathbf{P}_1^{H-1}\right.\\
    &\left.+\mathbf{P}_4\mathbf{P}_3^{D_m-\hat{H}-1}\mathbf{P}_2^{\hat{H}-H}\mathbf{P}_1^{H-1}\right)_{2\cdot}\bm{\mu}_1,
\end{split}
\end{equation}
\begin{equation}
    \begin{split}
        A^{\pi(H_0;H_1)}&=\left(\sum_{i=0}^{\hat{H}-H}\mathbf{P}_2^i\mathbf{P}_{1}^{H-1}\right)_{1\cdot}\bm{\mu}_1\\
    &+\left\|\left(\sum_{i=0}^{D_m-\hat{H}-1}\mathbf{P}_3^i\mathbf{P}_2^{\hat{H}-H}\mathbf{P}_1^{H-1}\right.\right.\\
    &\left.\left.+\mathbf{P}_4\mathbf{P}_3^{D_m-\hat{H}-1}\mathbf{P}_2^{\hat{H}-H}\mathbf{P}_1^{H-1}\right)\bm{\mu}_1\right\|_1.
    \end{split}
\end{equation}
For $H_0< H_1=D_m$, we can write
\begin{equation}
\begin{split}
    \mu_{1,0}&=(1-\lambda)p\left(\sum_{i=0}^{D_m-H-1}\mathbf{P}_2^i\mathbf{P}_{1}^{H-1}\right.\\
    &\left.+\mathbf{P}_4\mathbf{P}_2^{D_m-H-1}\mathbf{P}_1^{H-1}\right)_{1\cdot}\bm{\mu}_1\\
    &+\gamma p\left(\mathbf{P}_4\mathbf{P}_2^{D_m-H-1}\mathbf{P}_1^{H-1}\right)_{2\cdot}\bm{\mu}_1,
\end{split}
\end{equation}
\begin{equation}
    \begin{split}
        A^{\pi(H_0;H_1)}&=\left(\sum_{i=0}^{D_m-H-1}\mathbf{P}_2^i\mathbf{P}_{1}^{H-1}\right)_{1\cdot}\bm{\mu}_1\\
        &+\left\|\left(\mathbf{P}_4\mathbf{P}_2^{D_m-H-1}\mathbf{P}_1^{H-1}\right)\bm{\mu}_1\right\|_1.
    \end{split}
\end{equation}
For $H_0<D_m$ with $H_1=D_m+1$, we have
\begin{equation}
\begin{split}
    &\mu_{1,0}=(1-\lambda) p\left(\sum_{i=0}^{D_m-H-1}\mathbf{P}_2^i\mathbf{P}_{1}^{H-1}\right.\\
    &\left.+\mathbf{P}_4\mathbf{P}_2^{D_m-H-1}\mathbf{P}_1^{H-1}\right)_{1\cdot}\bm{\mu}_1,
\end{split}
\end{equation}
and $A^{\pi(H_0;H_1)}=\frac{\mu_{1,0}}{(1-\lambda) p}$.
For $H_0=D_m$ with $H_1=D_m+1$, we have
\begin{equation}
\begin{split}
    &\mu_{1,0}=(1-\lambda) p\left(\mathbf{P}_4\mathbf{P}_1^{D_m-2}\right)_{1\cdot}\bm{\mu}_1,
\end{split}
\end{equation}
and $A^{\pi(H_0;H_1)}=\frac{\mu_{1,0}}{(1-\lambda) p}$.

\section{Useful Expressions for Implementing Algorithm 2 When Channel is Error-free}\label{Expression_No_Error_Channel}
Consider a sub-MDP applying policy $\pi(H_0;H_1)$ with $\lambda$, $\gamma$, $D_m$, and $p=1$.
By the structure of the sub-MDP shown in Fig. \ref{MCTHnoError}, we have the following expressions.

Consider $H_1\le H_0$.
We can derive that $A^{\pi(H_0,H_1)}=\mu_{1,0}+\mu_{1,1}$.
Let $\Phi=\frac{(1-\lambda-\gamma)^{H-1}}{\lambda+\gamma}$.
For $H_1\le H_0\le D_m$, the linear equations can be expressed as
\begin{equation}\label{LP1CF}
    \begin{split}
        &\left[H+\frac{1-(1-\lambda)^{\hat{H}-H}}{\lambda}\left(\Phi\lambda+\frac{\gamma}{\lambda+\gamma}\right)\right]\mu_{1,0}\\
        &+\left[H+\frac{1-(1-\lambda)^{\hat{H}-H}}{\lambda}\left(-\Phi\gamma+\frac{\gamma}{\lambda+\gamma}\right)\right]\mu_{1,1}=1,
    \end{split}
\end{equation}
\begin{equation}\label{LP2CF}
    \begin{split}
        \mu_{1,0}&=(1-\lambda-\gamma)(1-\lambda)^{\hat{H}-H}\left(\Phi\lambda+\frac{\gamma}{\lambda+\gamma}\right)\mu_{1,0}\\
    &\!+\!(1\!-\!\lambda\!-\!\gamma)(1\!-\!\lambda)^{\hat{H}-H}\left(-\Phi\gamma\!+\!\frac{\gamma}{\lambda+\gamma}\right)\mu_{1,1}\\
    &+\gamma(\mu_{1,0}+\mu_{1,1}),
    \end{split}
\end{equation}
and 
\begin{equation}
    \begin{split}
        J^{\pi(H_0,H_1)}&\!=\!\left[\frac{1\!-\!(1\!-\!\lambda\!-\!\gamma)^{H+1}}{(\lambda+\gamma)^2}\!-\!\frac{(H+1)(1\!-\!\lambda\!-\!\gamma)^{H}}{\lambda\!+\!\gamma}\right]\\
        &\times\frac{-\lambda\mu_{1,0}+\gamma\mu_{1,1}}{\lambda+\gamma}
        +\frac{H(H+1)}{2}\frac{\lambda(\mu_{1,0}+\mu_{1,1})}{\lambda+\gamma}\\
        &+\left[H\!-\!(\hat{H}\!+\!1)(1\!-\!\lambda)^{\hat{H}\!-\!H}\!+\!\frac{1\!-\!(1\!-\!\lambda)^{\hat{H}\!-\!H\!+\!1}}{\lambda}\right]\\
        &\times \!\frac{(1\!-\!\lambda\!-\!\gamma)^{H\!-\!1}(\lambda\mu_{1,0}\!-\!\gamma\mu_{1,1})\!+\!\gamma(\mu_{1,0}\!+\!\mu_{1,1})}{\lambda+\gamma}.
    \end{split}
\end{equation}
For $H_1\le D_m$ while $H_0=D_m+1$, the linear equations are given by
\begin{equation}
\begin{split}
    &\left[H+\frac{1}{\lambda}\left(\Phi\lambda+\frac{\gamma}{\lambda+\gamma}\right)\right]\mu_{1,0}\\
    &+\left[H+\frac{1}{\lambda}\left(-\Phi\gamma+\frac{\gamma}{\lambda+\gamma}\right)\right]\mu_{1,1}=1,
\end{split}
\end{equation}
\begin{equation}
    \mu_{1,0}=\gamma(\mu_{1,0}+\mu_{1,1}),
\end{equation}
and
\begin{equation}
    \begin{split}
        &J^{\pi(H_0,H_1)}\\
        &=\left[\frac{1-(1-\lambda-\gamma)^{H+1}}{(\lambda+\gamma)^2}-\frac{(H+1)(1-\lambda-\gamma)^{H}}{\lambda+\gamma}\right]\\
        &\times\frac{-\lambda\mu_{1,0}+\gamma\mu_{1,1}}{\lambda+\gamma}
        +\frac{H(H+1)}{2}\frac{\lambda(\mu_{1,0}+\mu_{1,1})}{\lambda+\gamma}\\
        &+\left[H+\frac{1-(1-\lambda)^{D_m-H}}{\lambda}\right]\\
        &\frac{(1-\lambda-\gamma)^{H-1}(\lambda\mu_{1,0}-\gamma\mu_{1,1})+\gamma(\mu_{1,0}+\mu_{1,1})}{\lambda+\gamma}
    \end{split}
\end{equation}
when $H_1<D_m$;
\begin{equation}
    \begin{split}
        &J^{\pi(H_0,H_1)}\\
        &=\left[\frac{1-(1-\lambda-\gamma)^{D_m}}{(\lambda+\gamma)^2}-\frac{D_m(1-\lambda-\gamma)^{D_m-1}}{\lambda+\gamma}\right]\\
        &\times\frac{-\lambda\mu_{1,0}+\gamma\mu_{1,1}}{\lambda+\gamma}
        +\frac{D_m(D_m-1)}{2}\frac{\lambda(\mu_{1,0}+\mu_{1,1})}{\lambda+\gamma}\\
        &+D_m(\mu_{1,0}+\mu_{1,1}),
    \end{split}
\end{equation}
when $H_1=D_m$.

On the other hand, consider the situation $H_0<H_1$.
Due to the symmetry between the structure of the sub-MDPs in this situation and those in the symmetrical situation, i.e., $H_1<H_0$, the linear equations can be derived by swapping the positions of $\lambda,\gamma$ and $\mu_{1,0},\mu_{1,1}$ in the linear equations of the corresponding symmetrical contexts, and $A^{\pi(H_0,H_1)}$ is identical to that in the symmetrical contexts.
In addition, for $H_0\le H_1\le D_m$, 
\begin{equation}
    \begin{split}
        &J^{\pi(H_0,H_1)}\\
        &=\left[\frac{1-(1-\lambda-\gamma)^{H+1}}{(\lambda+\gamma)^2}-\frac{(H+1)(1-\lambda-\gamma)^{H}}{\lambda+\gamma}\right]\\
        &\times\frac{\gamma\mu_{1,1}-\lambda\mu_{1,0}}{\lambda+\gamma}
        +\frac{H(H+1)}{2}\frac{\lambda(\mu_{1,0}+\mu_{1,1})}{\lambda+\gamma}\\
        &+\left[\frac{H(1-\gamma)}{\gamma}-\frac{(\hat{H}+1)(1-\gamma)^{\hat{H}-H+1}}{\gamma}\right.\\
        &\left.+\frac{(1-\gamma)(1-(1-\gamma)^{\hat{H}-H+1})}{\gamma^2}\right]\\
        &\times \frac{(1-\lambda-\gamma)^{H-1}(\gamma\mu_{1,1}-\lambda\mu_{1,0})+\lambda(\mu_{1,0}+\mu_{1,1})}{\lambda+\gamma}.
    \end{split}
\end{equation}
For $H_0\le D_m$ while $H_1=D_m+1$, we have
\begin{equation}
    \begin{split}
        &J^{\pi(H_0,H_1)}\\
        &=\left[\frac{1-(1-\lambda-\gamma)^{H+1}}{(\lambda+\gamma)^2}-\frac{(H+1)(1-\lambda-\gamma)^{H}}{\lambda+\gamma}\right]\\
        &\times\frac{\gamma\mu_{1,1}-\lambda\mu_{1,0}}{\lambda+\gamma}
        +\frac{H(H+1)}{2}\frac{\lambda(\mu_{1,0}+\mu_{1,1})}{\lambda+\gamma}\\
        &+\left[\frac{H(1-\gamma)}{\gamma}+\frac{(1-\gamma)(1-(1-\gamma)^{D_m-H})}{\gamma^2}\right]\\
        &\times \frac{(1-\lambda-\gamma)^{H-1}(\gamma\mu_{1,1}-\lambda\mu_{1,0})+\lambda(\mu_{1,0}+\mu_{1,1})}{\lambda+\gamma},
    \end{split}
\end{equation}
when $H_0<D_m$, and
\begin{equation}
    \begin{split}
        &J^{\pi(H_0,H_1)}\\
        &=\left[\frac{1-(1-\lambda-\gamma)^{D_m}}{(\lambda+\gamma)^2}-\frac{D_m(1-\lambda-\gamma)^{D_m-1}}{\lambda+\gamma}\right]\\
        &\times\frac{\gamma\mu_{1,1}-\lambda\mu_{1,0}}{\lambda+\gamma}
        +\frac{D_m(D_m-1)}{2}\frac{\lambda(\mu_{1,0}+\mu_{1,1})}{\lambda+\gamma}\\
        &+\frac{D_m}{\gamma}\frac{(1-\lambda-\gamma)^{H-1}(\gamma\mu_{1,1}-\lambda\mu_{1,0})+\lambda(\mu_{1,0}+\mu_{1,1})}{\lambda+\gamma}.
    \end{split}
\end{equation}
For $H_0=H_1=D_m+1$, $J^{\pi(H_0;H_1)}=\frac{D_m\lambda}{\lambda+\gamma}$, and $A^{\pi(H_0;H_1)}=0$, which are identical to that when $p<1$.

\end{document}